\documentclass[pdflatex,sn-mathphys-num]{sn-jnl}

\usepackage{graphicx,soul}%
\usepackage{epstopdf}
\usepackage{a4wide}
\usepackage{multirow}%
\usepackage{amsmath,amssymb,amsfonts}%
\usepackage{amsthm}%
\usepackage{mathrsfs}%
\usepackage[title]{appendix}%
\usepackage{xcolor}%
\usepackage{textcomp}%
\usepackage{manyfoot}%
\usepackage{booktabs}%
\usepackage{algorithm}%
\usepackage{algorithmicx}%
\usepackage{algpseudocode}%
\usepackage{listings}%
\usepackage{cancel}
\usepackage{enumitem}
\usepackage{thmtools}
\usepackage{thm-restate}

\theoremstyle{thmstyleone}%

\theoremstyle{thmstyletwo}%
\newtheorem{remark}{Remark}%

\theoremstyle{thmstylethree}%
\newtheorem{definition}{Definition}%

\begin{document}

\title[Article Title]{ Modeling the Impact of Immune Boosting on Population-Level Vaccine Effectiveness}

\author*[1]{\fnm{Nir} \sur{Gavish}}\email{ngavish@technion.ac.il}
\author[2]{\fnm{Guy} \sur{Katriel}}\email{katriel@braude.ac.il}
\author[1]{\fnm{Zohar} \sur{Rom-Ramon}}\email{rom.zohar@campus.technion.ac.il}
\affil*[1]{\orgdiv{Faculty of Mathematics}, \orgname{Technion Israel Institute of Technology}, \orgaddress{\street{}, \city{Technion City, Haifa}, \postcode{3200003}, \state{}, \country{Israel}}}

\affil[2]{\orgdiv{Department of Applied Mathematics}, \orgname{Braude College of Engineering}, \orgaddress{\street{}, \city{Karmiel}, \postcode{2161002}, \state{}, \country{Israel}}}

\abstract{We extend the standard susceptible-infected-recovered framework to incorporate natural immune boosting during a short-scale outbreak. By deriving closed-form final size relations, we analytically link total attack rates to boosting dynamics and vaccine coverage. This framework identifies a critical boosting threshold: above it, higher vaccine coverage paradoxically decreases relative vaccine effectiveness. This occurs because successful epidemic suppression deprives vaccinated individuals of the silent pathogen exposures required to maintain their relative immunological advantage. Crucially, the overall population-level impact remains beneficial, consistently reducing absolute disease burden. For highly transmissible variants, asymptotic analysis reveals that relative vaccine effectiveness converges to a positive limit entirely independent of coverage.

\textbf{Relevance to the Life Sciences.} Natural immune boosting, e.g., via low-level or asymptomatic exposures, reinforces individual immunity. We clarify how this biological mechanism translates to population-level metrics, focusing on its impact on the observed effectiveness of leaky vaccines during a single epidemic outbreak. Our key biological finding is that successful vaccination campaigns can inadvertently suppress the very exposures required to sustain a vaccinated cohort's relative advantage. Consequently, public health metrics that ignore boosting risk misinterpreting a successful campaign's lower relative effectiveness as biological vaccine failure. To prevent this epidemiological misinterpretation, surveillance strategies could be adapted to incorporate serial serological surveys and contact tracing to track silent booster accumulation.

\textbf{Mathematical Content.} We introduce an immune-boosted compartment into a deterministic compartmental model featuring leaky vaccination and perfect boosting via regular and weak pathogen exposures. Rather than solving implicit systems, we derive explicit, closed-form final size equations as a function of the cumulative force of exposure. This formulation enables rigorous analysis of vaccine effectiveness across boosting regimes, identifying a critical boosting threshold, fundamentally distinct from the standard epidemic threshold, that dictates the behavior of population-level vaccine metrics. Additionally, a probabilistic framework based on the Poisson distribution isolates multiple exposure effects, validating threshold limits and proving that relative effectiveness decouples from coverage in highly transmissible variants.}


\maketitle 

\section{Introduction}

Vaccination has stood for over two centuries as one of the most effective strategies for controlling infectious diseases, successfully suppressing endemic threats such as measles and polio \cite{plotkina1999vaccination}. The recent global response to the COVID-19 pandemic further underscored this pivotal role, demonstrating how rapid vaccine deployment can fundamentally alter the trajectory of emerging epidemic outbreaks \cite{burki2022omicron,tregoning2021progress}. A central challenge in the design of public health policy is the complex link between individual-level vaccine protection and population-level vaccine impact \cite{ halloran1997study, shim2012distinguishing}. This translation is non-trivial: the population level impact of vaccination is shaped not only by the direct biological protection conferred by the vaccine but also by indirect protection and herd immunity. Mathematical models are indispensable tools for bridging this gap \cite{kretzschmar2019disease}.

Individual vaccine protection is rarely absolute or static. In some cases, protection is subject to temporal waning, as neutralizing antibody titers naturally decay in the months or years following vaccination \cite{menegale2023evaluation, vashishtha2024durability}. Furthermore, a person's baseline immunity is continuously modulated by their unique history of pathogen exposures. For example, re-exposure to a circulating pathogen can efficiently re-stimulate immune memory, thus lowering an individual’s susceptibility to future disease.  Notably, such natural immune boosting can be induced by silent exposures that do not result in symptomatic infection or transmission \cite{barbarossa2017stability,barbarossa2015immuno,barbarossa2015mathematical,dafilis2012influence,lavine2013immune, lavine2011natural,leung2019models}.  

The epidemiology of pertussis serves as a compelling illustration of how immune boosting can significantly alter the population-level implications of vaccination campaigns. Following the introduction of routine vaccination in the 1950s, incidence rates in Massachusetts dropped significantly, suggesting near-elimination after two decades. Yet, despite sustained high vaccination coverage, the disease unexpectedly resurged in the 1980s. This phenomenon was also observed in other countries, prompting investigations that led to several proposed explanations \cite{aguas2006pertussis,lavine2013immune,lavine2011natural, mooi2001adaptation, mooi2009bordetella,yih2000increasing}.
The authors of \cite{lavine2011natural} attribute this effect to natural immune boosting. In the pre-vaccine era, frequent natural exposure to pertussis acted as a ``booster”, maintaining strong personal immunity. However, in the vaccine era, decreased disease circulation led to the elimination of natural boosters, so that immunity faded before individuals were re-exposed.
The study \cite{lavine2011natural} highlighted the role of weak exposures, those insufficient to cause infection, in triggering a boosting response, alongside regular exposures which may either cause infection or boost protection. They deduced this by observing that the probability of experiencing boosting is significantly higher than that of a primary infection.
 
Existing theoretical studies have provided valuable insights into how immune boosting affects the long-term dynamics of endemic diseases. Typically, these studies examine immune boosting in the context of waning immunity, assuming a mean duration of immunity ranging from a couple of years to a decade or more. For instance, \cite{lavine2013immune,lavine2011natural} investigated the impact of immune boosting on long-term pertussis dynamics, while  \cite{barbarossa2017stability, barbarossa2015immuno,barbarossa2015mathematical, dafilis2012influence} developed general mathematical frameworks for modeling the interplay between waning immunity and immune boosting.

Shifting from endemic scenarios to single outbreaks, a recent study by Park et al. \cite{park2023immune} incorporated immune boosting of vaccinated individuals into an SIR model. By focusing on a short-term outbreak, their model naturally omits waning immunity and vital dynamics, aiming instead to bridge different vaccination models through the boosting mechanism. Building on this foundation, our focus turns to the largely unexplored impact of immune boosting on population-level protection during single, short-scale epidemics.

In this work, we aim to quantify the effects of immune boosting on the dynamics of a single epidemic occurring over a short time scale, so that waning immunity and vital dynamics can be neglected. Building on the framework established by \cite{park2023immune}, we investigate an extended SIR model with a leaky vaccine to understand the key population-level effects of immune boosting. 

Similarly to \cite{park2023immune}, we assume perfect boosting: vaccinated individuals who get exposed and remain uninfected can gain full and permanent immunity. Based on the findings from \cite{lavine2011natural}, our model distinguishes between two types of encounters: regular exposures, which are sufficient to cause infection, and weak exposures, which cannot cause infection but can still lead to boosting. 
Ultimately, we aim to answer the question: how does immune boosting influence population-level vaccine effectiveness in a single outbreak?

To address this question, we derive closed-form final size relations that explicitly link the attack rates to the boosting parameter, vaccine coverage, and the cumulative force of exposure, an approach that allows for a simplified analysis without the need to solve implicit final size equations. This analytical framework enables us to quantify both overall population-level protection and the relative protection of vaccinated versus unvaccinated individuals. 

Using analytical and numerical methods, we show that while immune boosting consistently enhances overall protection, relative vaccine effectiveness (${\rm VE}$), that is the protection gained by vaccinated individuals relative to unvaccinated individuals, depends in a more complex way on the model's parameters. 
We establish that under infrequent boosting, relative vaccine effectiveness declines as the total scale of the outbreak (cumulative force of exposure) increases, similar to the scenario with no boosting \cite{halloran1999design}. However, we show that under sufficiently strong immune boosting, effectiveness increases with the cumulative force of exposure, as vaccinated individuals disproportionately benefit from these silent, immunity-reinforcing encounters. This mechanism drives a {\em vaccine effectiveness paradox}: when boosting strength is above a threshold, increasing vaccine coverage paradoxically reduces  relative ${\rm VE}$. Higher coverage suppresses the epidemic, thereby depriving vaccinated individuals of the frequent natural exposures needed to fully maintain their relative advantage. However, as expected, the absolute protection, represented by ${\rm VE}_{\rm overall}$, still increases with increasing vaccine coverage. In addition, in classic SIR models with very high reproduction numbers, assuming a leaky vaccine, virtually the entire population becomes infected, causing vaccine effectiveness to drop to zero. Here, we demonstrate that when immune boosting is present, relative protection against highly transmissible variants converges to a positive limit that is entirely independent of vaccine coverage. We elucidate these findings using a  probabilistic framework. This approach recovers the critical boosting threshold, offers an intuitive explanation for the vaccine effectiveness ``paradox”, and captures the asymptotic behavior of $\rm{VE}$ for highly transmissible pathogens.

The remainder of this paper is organized as follows. Section \ref{sec:math_model} introduces the extended SIR model with vaccination and immune boosting and derives final size equations. In Section \ref{sec:VE}, we analyze the behavior of vaccine effectiveness, present the vaccine effectiveness ``paradox", and study the asymptotic behavior of highly transmissible variants.
Section \ref{sec:approximateVE} then introduces an individual-level probabilistic framework to provide a clear mechanistic explanation for these counterintuitive results. Finally, in Section \ref{sec:discussion}, we provide concluding remarks and discuss future directions.

\section{Compartmental Model}\label{sec:math_model}
We formulate a compartmental SIR transmission model to investigate the population-level dynamics of immune boosting during a single epidemic wave in a closed population. Due to the short time scale of the outbreak,  we neglect demographic turnover (births and deaths) and the waning of immunity among recovered or vaccinated individuals. In this model, we incorporate a leaky vaccine, which provides vaccinated individuals with partial protection by reducing their susceptibility. Furthermore, we assume all vaccines are administered initially, and there is no continuous vaccination during the pandemic wave.
Thus, the population is divided into two primary groups: vaccinated and unvaccinated. 
The state variables $S_i, I_i,$ and $R_i$ denote the fractions of the population that are susceptible, infected, and recovered, respectively, where the subscript $i \in \{u, v\}$ indicates the vaccination status. 

Beyond the standard SIR model, we consider the possibility of immune boosting, where pathogen exposure can trigger an immune response among susceptible vaccinated individuals who are exposed but do not become infected. In our model, boosting is defined as an event in which these susceptible individuals acquire robust, non-waning immunity. Therefore, they become equivalent to those in the $R_v$ compartment. To better reveal the impact of boosting, we define a dedicated compartment, $R_b$, for the vaccinated individuals who have undergone such immune boosting. 
The model distinguishes between two types of pathogen exposures:
\begin{itemize}
    \item Regular Exposures: These represent typical exposures with the pathogen that can lead to infection in both vaccinated and unvaccinated individuals. These exposures also exist in standard SIR models. However, in our framework, these exposures can trigger immune boosting if a vaccinated person is exposed but does not become infected.

    \item Weak Exposures: These encounters involve low pathogen loads that are insufficient to cause infection. Because they never result in illness, they are entirely overlooked by classical SIR models. Nevertheless, following \cite{lavine2011natural}, we account for the possibility that these exposures can induce an immune boosting response, conferring immunity against future exposures to the pathogen.
\end{itemize}
The model equations are as follows:
\begin{subequations}\label{eq:Model}
\begin{align}
S_u^\prime(t)&=-s\lambda(t)S_u(t)\label{eq:Su},\\
I_u^\prime(t)&=s\lambda(t)S_u(t)-\gamma_uI_u(t)\label{eq:Iu},\\
R_u^\prime(t)&=\gamma_uI_u(t)\label{eq:ru},\\
S_v^\prime(t)&=\rm{\underbrace{-s\epsilon\lambda(t)S_v(t)}_{Infection}}-\underbrace{sq_s(1-\epsilon)\lambda(t)S_v(t)}_{Boosting- regular\, exposure}-\underbrace{q_w(1-s)\lambda(t)S_v(t)}_{Boosting- weak\, exposure}\label{eq:Sv}\\
&=-(\delta+s\epsilon)\lambda(t)S_v(t) \nonumber,\\
I_v^\prime(t)&=s\epsilon\lambda(t)S_v(t)-\gamma_vI_v(t)\label{eq:Iv},\\
R_v^\prime(t)&=\gamma_vI_v(t)\label{eq:rv},\\
R_b^\prime(t)&=\underbrace{sq_s(1-\epsilon)\lambda(t)S_v(t)}_{\rm Boosting-regular\, exposure}+\underbrace{q_w(1-s)\lambda(t)S_v(t)}_{\rm Boosting- weak\, exposure}=\delta\lambda(t)S_v(t),\label{eq:R_b}
\end{align}
where the total interaction rate (the force of exposure), $\lambda(t)$, takes the form
\begin{equation}\label{eq:lambda}
    \lambda(t)=\beta_uI_u(t)+\beta_vI_v(t),
\end{equation}
see Figure \ref{fig:diagramofmodel} for the compartmental flow of the model.
We consider the initial conditions
\begin{equation}\label{eq:InitialConditions}  
\begin{split}
   &S_v(0)=\phi-I_v(0),\quad S_u(0)=1-\phi-I_u(0),\\
&I_u(0)+I_v(0)>0,\\
&R_v(0)=R_u(0)=R_b(0)=0, 
\end{split} 
\end{equation}
where $\phi \in (0,1)$ is the proportion of individuals that are vaccinated prior to the epidemic. The model assumes that the population mixes homogeneously, and that the initial population is normalized to one,
\begin{equation}\label{eq:notmalizedIC}
    \begin{split}
    S_u(0)+S_v(0)+I_u(0)+I_v(0)+R_u(0)+R_v(0)+R_b(0)=1.
    \end{split}
\end{equation}
\end{subequations}
\begin{figure}[ht!]
    \centering
\includegraphics[width=0.8\linewidth]{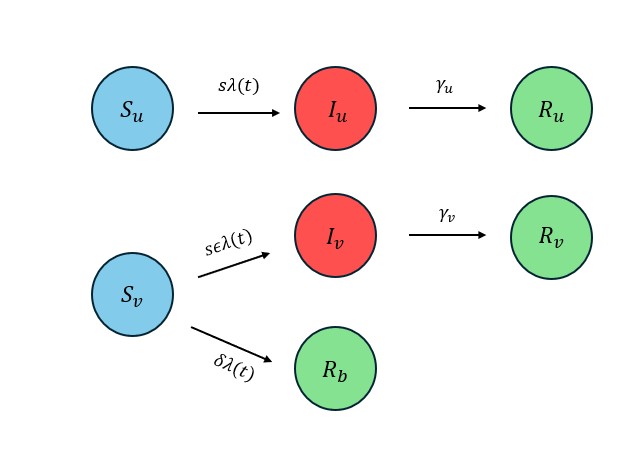}
    \caption[Diagram of disease dynamics]
    {Flow diagram of the compartmental SIR model incorporating immune boosting. The model tracks two parallel populations: unvaccinated (top) and vaccinated (bottom). Unvaccinated susceptibles ($S_u$) transition to the infected class ($I_u$) at a rate $s\lambda(t)$. For vaccinated susceptibles ($S_v$), the leaky vaccine reduces the infection rate to $s\epsilon\lambda(t)$. $S_v$ individuals who are exposed to the pathogen but do not become infected, either through regular or weak exposures, transition directly to the immune-boosted compartment ($R_b$) at a rate $\delta\lambda(t)$, acquiring full and permanent immunity.}
    \label{fig:diagramofmodel}
\end{figure}
The model parameters, along with their epidemiological definitions and ranges, are summarized in Table~\ref{tab:parameters}.
The parameters $\beta_i$ describe the transmission coefficients and $\gamma_i$ describe the recovery rates for $i\in \{u, v\}$. We consider $\beta_i, \gamma_i>0$.
The parameter $\epsilon\in (0,1)$ represents the relative susceptibility of vaccinated individuals, so that $1-\epsilon$ is the vaccine efficacy against infection in vaccinated individuals.
\begin{table}[htbp]
\centering
\caption{Summary of Model Parameters}
\label{tab:parameters}
\small
\begin{tabular}{lp{8.5cm}}
\hline
\textbf{Parameter} & \textbf{Description} \\ \hline
$\beta_i$ & Transmission coefficients ($i \in \{u, v\}$) \\
$\gamma_i$ & Recovery rates ($i \in \{u, v\}$) \\
$\epsilon$ & Relative susceptibility of vaccinated individuals \\
$1-\epsilon$ & Vaccine efficacy against infection \\
$s$ & Proportion of regular exposures among all exposures \\
$q_s$ & probability of boosting following regular exposure\\
$q_w$ & Probability of boosting following weak exposure \\
$\delta$ & Composite effective boosting parameter \\ \hline
\end{tabular}
\end{table}

We define $s\in(0, 1]$ as the proportion of regular exposures among all types of exposure, and $q_s$ as the probability of experiencing boosting following a regular exposure that does not result in infection. Weak exposures occur with probability $1-s$, and the probability they will lead to boosting among vaccinated individuals is $q_w$.
The parameters $q_s$, $q_w$ appear in the model through the composite parameter $\delta$:
\begin{equation}\label{eq:delta}
    \delta=(1-\epsilon)sq_s+(1-s)q_w.
\end{equation}
Since $q_s,q_w \in [0,1]$, $\delta \in[0,1-s\epsilon]$.
The parameter $\delta$ serves as the effective {\em boosting parameter}, capturing the combined impact of boosting induced by regular and weak exposures.  Since weak exposures cannot cause disease, vaccine efficacy is measured exclusively against regular encounters. Thus, $\epsilon$ represents the conditional probability that a regular exposure leads to infection.

Note that by setting $s = 1$ and $q_s = 0$ in \eqref{eq:Model}, we fully reduce to the classic SIR model without boosting. Similarly, if we consider $q_s = q_w = 0$ and $s < 1$, which implies $\delta = 0$, we obtain the classic SIR model with a scaled transmission rate, $\tilde{\beta}_i = s\beta_i$. In this scenario, the probability of immune boosting is zero. In addition, $\delta=1-s\epsilon$ if and only if $q_s=q_w=1$. In this case, every exposure leads to infection or boosting, so that a vaccinated susceptible never remains susceptible following exposure.

A standard calculation using the next generation matrix method  (see Appendix \ref{secA:Rv}) yields the reproduction number, $R_v$:  
\begin{equation}\label{eq:Rv}
    \mathcal{R}_v=s\left(\frac{\beta_u (1-\phi)}{\gamma_u}+\frac{\epsilon\beta_v \phi}{\gamma_v} \right). 
\end{equation}  
It is important to note that $\mathcal{R}_v$ is not affected by the boosting mechanism, and is therefore independent of the boosting parameters $q_s$ and $q_w$. This implies that the herd immunity threshold, the minimal vaccine coverage that will prevent an epidemic, is not affected by immune boosting. However, our analysis reveals that boosting does impact the final size of the epidemic.

Our model accounts for possible differences between the unvaccinated and vaccinated transmission rates, $\beta_u$ and $\beta_v$, which may stem from both biological and behavioral factors. Biologically, vaccines typically reduce viral load and may shorten the infectivity period. Therefore, our primary focus is on the scenario where 
\begin{equation}\label{eq:vaccineassumption}
    \frac{\beta_v}{\gamma_v} \leq \frac{\beta_u}{\gamma_u}.
\end{equation}
We will also address risk compensation scenarios for which behavioral shifts, such as unmasking or maintaining regular social routines due to milder infections, can lead to $\beta_v > \beta_u$.  In all such cases, we restrict our analysis to the regime in which the vaccine remains net-beneficial, satisfying
\begin{equation}\label{eq:net_beneficial}
\frac1s\mathcal{R}_v^\prime(\phi)=\frac{\epsilon\beta_v}{\gamma_v}-\frac{\beta_u}{\gamma_u}\le 0.    
\end{equation}

\subsection{Model Simulations}\label{sec:modelsimulations}

\begin{figure}[ht!]
    \centering
    \includegraphics[width=\linewidth]{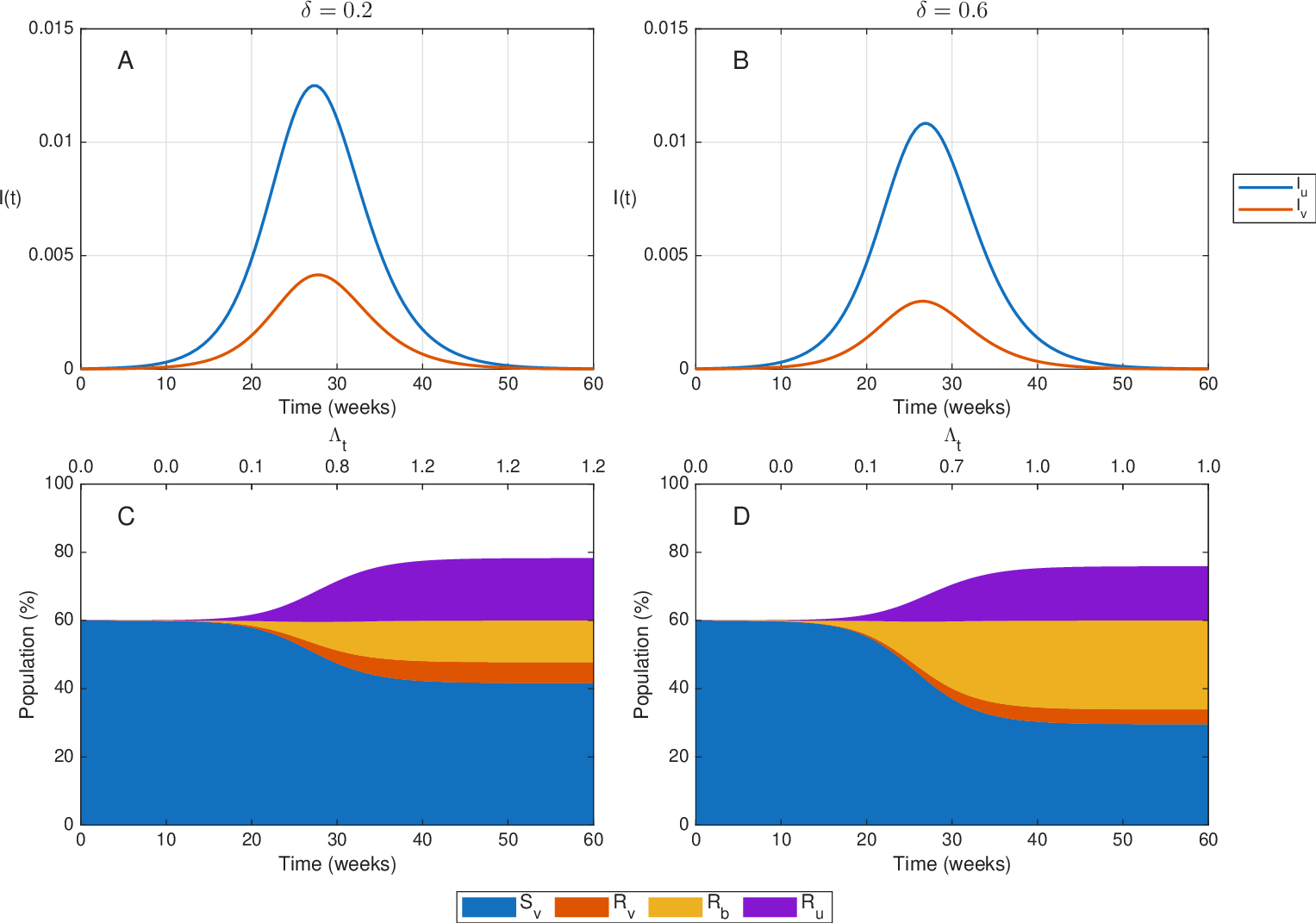}
   \caption{{\bf A.} Time evolution of the $I_u$ and $I_v$ populations for $\delta = 0.2$ (left) and $\delta = 0.6$ (right).
   {\bf B.} Population changes across compartments $S_v, R_v, R_b$, and $R_u$. The left panel simulates a infrequent boosting regime ($\delta = 0.2$), while the right panel simulates a frequent boosting regime ($\delta = 0.6$). Fixed parameters for both simulations are $\beta_u=\beta_v=5$, $\gamma_u=\gamma_v=1$, $\phi = 0.6, \epsilon = 0.2$, and $s = 0.5$. The secondary upper axis indicates the corresponding cumulative force of exposure, $\Lambda_t$.}
    \label{fig:modelsimulations}
\end{figure}
The focus of this work is on the cumulative burden of disease with related  measures such as attack rates and vaccine effectiveness. However, we begin by numerically demonstrating the dynamics of~\eqref{eq:Model}, under regimes of infrequent ($\delta = 0.2$) and frequent ($\delta = 0.6$) immune boosting. For simplicity, we consider a case where both unvaccinated and vaccinated populations share the same transmission coefficient, $\beta_u = \beta_v = 5$/week and  recovery rate $\gamma = 1$/week. Additionally, the proportion of regular interactions that can lead to infection is $s = 0.5$. We consider an initial vaccination coverage of $\phi = 0.6$ with a vaccine efficacy of $1-\epsilon=0.8$. The initial conditions for the infected compartments are $I_u(0)=I_v(0)=10^{-5}$.

Figure~\ref{fig:modelsimulations} illustrates these dynamics. Panel A portrays the size of the infected populations as a function of time within both the unvaccinated and vaccinated compartments, demonstrating the single outbreak dynamics. As expected, Panel B shows that an increased boosting rate, parametrized by~$\delta$, results in a larger proportion of the population transitioning into the boosted compartment, $R_b$, without being infected. While this initial observation is straightforward, the following section reveals unexpected differences in behavior between infrequent ($\delta=0.2$) and frequent ($\delta=0.6$) boosting regimes.

Following \cite{breda2012formulation, sellke1983asymptotic, diekmann2021discrete}, 
we also present epidemic progression as a function of
the cumulative force of exposure,
\begin{equation}\label{eq:lambdat}
\Lambda_t = \int_0^t \lambda(\tau) \, d\tau,
\end{equation}
where $\lambda(t)$ is the force of exposure \eqref{eq:lambda}. This metric is tracked on the secondary upper axis of Figure \ref{fig:modelsimulations}. We observe that $\lim_{t\to\infty}\Lambda_t\approx 1.2$ for $\delta=0.2$ and drops to $\lim_{t\to\infty}\Lambda_t\approx 1$ for $\delta=0.6$, indicating a reduction in the total scale of the outbreak due to immune boosting. In the following section, we demonstrate the utility of analyzing cumulative disease measures as a function of the cumulative force of exposure over the whole course of the epidemic.

\subsection{Final Size Relations}\label{sec:finalsizeLambda}
The cumulative burden of disease is characterized by the attack rates for the unvaccinated ($AR_u$) and vaccinated ($AR_v$) populations, which describe the relative accumulated number of infections in each population.
\begin{definition}\label{def:attackrates}
The attack rates in the unvaccinated and vaccinated population are defined respectively as
\begin{equation}\label{eq:ARdef}
  AR_u=  \frac{R_u(\infty)}{S_u(0)},\quad AR_v= \frac{R_v(\infty)}{S_v(0)}.   
\end{equation}
\end{definition}
We calculate the attack rates as a function of the cumulative force of exposure, \eqref{eq:lambdat}, over the whole course of the epidemic
\begin{equation}\label{eq:Lambdadef}
\Lambda:=\lim_{t\to\infty}\Lambda_t=\int_{0}^{\infty} {\lambda(\tau)\,d\tau}.
\end{equation}
\begin{restatable}{proposition}{FinalSizeLambda}\label{prop:finalsize}
Consider a solution to the system \eqref{eq:Model}. 

 In the limit $\left.I_u(0),I_v(0)\to 0\right.$,  the attack rates, $AR_u$ and $AR_v$, in the unvaccinated and vaccinated populations, are given by: \begin{subequations}\label{eq:FinalSizeLambda}
    \begin{align}
     &AR_u=1-\exp{(-s\Lambda)}\label{eq:finalsizeULambda},\\
     &AR_v=\frac{s\epsilon}{\delta+s\epsilon}(1-\exp{\left(-(\delta+s\epsilon) \Lambda\right)}),\label{eq:finalsizeVLambda}
    \end{align}   
     where the cumulative force of exposure, $\Lambda$, is the solution of the equation:
     \begin{equation}
     \Lambda= \frac{\beta_u(1-\phi)}{\gamma_u}(1-\exp{(-s\Lambda)})+\frac{\beta_v\phi s\epsilon}{(\delta+s\epsilon)\gamma_v}\left(1-\exp{(-(\delta+s\epsilon)\Lambda)}\right).\label{eq:finalsizelambda}
    \end{equation}
\end{subequations}
\end{restatable}
\begin{proof}
    See Appendix \ref{sec:FinalsizeappLambda}.
\end{proof}
Since \eqref{eq:finalsizeULambda}, \eqref{eq:finalsizeVLambda} expresses $AR_u$ and $AR_v$ as a function of $\Lambda$, we characterize the conditions under which \eqref{eq:finalsizelambda} has a unique solution.

\begin{restatable}{proposition}{existenceLambda} \label{prop:existenceLambda}
 Equation \eqref{eq:finalsizelambda} has a unique positive solution if $\mathcal{R}_v>1$, and no positive solution if $\mathcal{R}_v\leq 1$. 
\end{restatable}
\begin{proof}
    See Appendix \ref{sec:UninqueLamdaapp}.
\end{proof}

\section{Vaccine Effectiveness}\label{sec:VE}
In this section, we investigate the population-level impact of vaccination, with a specific focus on the effects of immune boosting. To do so, we employ the standard metric of population-level vaccine effectiveness, $\mathrm{VE}$, defined as \cite{farrington1993estimation,halloran1999design, orenstein1985field,shim2012distinguishing}:
\begin{equation}\label{eq:VE}
    \mathrm{VE} := 1 - \frac{AR_v}{AR_u},
\end{equation}
where $AR_v$ and $AR_u$ are the attack rates among vaccinated and unvaccinated individuals, respectively, see Definition~\eqref{def:attackrates}. 
Vaccine effectiveness quantifies the protection that vaccinated individuals receive over the course of the entire epidemic relative to unvaccinated individuals. From the perspective of a potential vaccinee, it addresses a fundamental question: What is the individual gain from vaccination? That is, how much does it reduce my chance of eventually becoming infected throughout the course of the epidemic?

In the absence of boosting, a highly contagious epidemic results in virtually all individuals becoming infected when the vaccine is leaky, regardless of vaccination status. Consequently, vaccine effectiveness approaches zero~\cite{halloran1999design}. More generally, we expect $\mathrm{VE}$ to decline as the total scale of the outbreak (cumulative force of exposure) increases. Therefore, since greater coverage suppresses epidemic spread, vaccine effectiveness is expected to increase with higher vaccine coverage.

In what follows, we demonstrate that the presence of immune boosting can invert these expected trends. Specifically, we establish that under infrequent boosting, vaccine effectiveness decreases as the total scale of the outbreak increases, as in the scenario with no boosting. In contrast, under stronger boosting conditions, vaccine effectiveness increases with epidemic spread. We call this the {\em vaccine effectiveness paradox}, wherein $\mathrm{VE}$ declines as vaccine coverage increases, even though the total epidemic burden continues to decrease. We formalize these findings below through a detailed analysis of $\mathrm{VE}$ under distinct boosting regimes.

\subsection{Analysis of~${\rm VE}$ - Distinct Boosting Regimes}\label{sec:VE_analysis}

Substituting Equations \eqref{eq:finalsizeULambda}, \eqref{eq:finalsizeVLambda} into~\eqref{eq:VE}, we obtain~${\rm VE}$ as an explicit function of the cumulative force of exposure $\Lambda$,
\begin{equation}\label{eq:VELambda}
    {\rm VE}(\Lambda)=1-\frac{s\epsilon}{\delta+s\epsilon}\frac{1-\exp(-(\delta+s\epsilon)\Lambda)}{1-\exp(-s\Lambda)}.
\end{equation}
In the classical SIR model without boosting, it is well-established that~${\rm VE}$ decays monotonically with~$\Lambda$ so that it is bounded by vaccine efficacy, ${\rm VE} < 1 - \epsilon$ \cite{halloran1999design}. This monotonicity and the related bound is confirmed by applying Equation \eqref{eq:VELambda} with $\delta = 0$ and $s = 1$, which yields:
\begin{equation*}
    {\rm VE}=1-\frac{1-\exp(-\epsilon\Lambda)}{1-\exp(-\Lambda)}.
\end{equation*}
Notably, while this behavior holds for the standard model, we demonstrate in the following section that it does not extend to scenarios involving sufficiently strong boosting.
\begin{restatable}{proposition}{VEanalyticMonotone} \label{prop:VEanalyticMonotone}
Let the vaccine effectiveness, ${\rm VE}(\Lambda)$, be defined as in Equation \eqref{eq:VELambda},
where $\Lambda > 0$ is the cumulative force of exposure, with parameters $s, \epsilon \in (0,1)$ and $\delta \in(0,1-s\epsilon)$.
Let
\[
\delta^*=(1 - \epsilon)s.
\]
Then, the monotonicity of ${\rm VE}(\Lambda)$ with respect to $\Lambda$ is determined by the value of $\delta$ relative to $\delta^*$:
\begin{enumerate}
    \item If $\delta < \delta^*$, ${\rm VE}(\Lambda)$ is a strictly decreasing function of $\Lambda$, and ${\rm VE}(\Lambda) < 1 - \epsilon$.
    \item If $\delta = \delta^*$, ${\rm VE}(\Lambda) = 1 - \epsilon$ for all $\Lambda>0$.
    \item If $\delta > \delta^*$, ${\rm VE}(\Lambda)$ is a strictly increasing function of $\Lambda$, and ${\rm VE}(\Lambda) > 1 - \epsilon$.
\end{enumerate}
\end{restatable} 

\begin{proof}
    See Appendix \ref{sec:VEproofAPP}.
    Note that the proof relies on Equation~\eqref{eq:VELambda}, rather than on the implicit final size relation in Equation \eqref{eq:finalsizelambda}.
\end{proof}
 Proposition \ref{prop:VEanalyticMonotone} reveals three distinct behavioral regimes for ${\rm VE}$, which depend entirely on the value of the boosting parameter $\delta$, as illustrated in Figure \ref{fig:MastercodeVE}A.
 In particular, when~$\delta>\delta^*$, see dotted yellow curve in Figure \ref{fig:MastercodeVE}A, ${\rm VE}$ increases with~$\Lambda$, which is directly opposite to its behavior in the classical SIR model and in the case of infrequent boosting,~$\delta<\delta^*$, see blue solid curve in Figure \ref{fig:MastercodeVE}A.
 \begin{figure}[ht!]
\includegraphics[width=\linewidth]{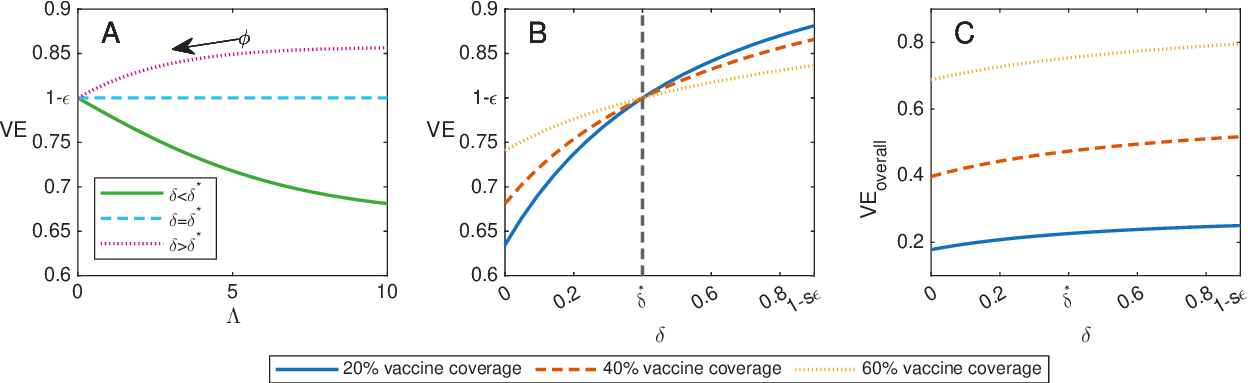}
   \caption{{\bf A.} Vaccine effectiveness (${\rm VE}$) as a function of the cumulative force of exposure ($\Lambda$), using the same parameters as in Figure \ref{fig:modelsimulations}. All curves intersect near $\Lambda = 0$, where they equal $1 - \epsilon$, and subsequently diverge depending on the value of $\delta$. The arrow for vaccine coverage ($\phi$) illustrates the ${\rm VE}$ paradox demonstrated in Proposition \ref{prop:vaccination_boosting}: when $\delta>\delta^*$, vaccine coverage reduces the cumulative force of exposure, thereby surprisingly decreasing vaccine effectiveness. {\bf B.} Vaccine effectiveness calculated as a function of the boosting parameter, $\delta$,  for $\phi=0.2, 0.4, 0.6$. All other parameters remain as in Figure \ref{fig:modelsimulations}$: \epsilon=0.2$, $\beta_u=\beta_v=5$, $\gamma_u=\gamma_v=1$, $s=0.5$. The corresponding reproduction numbers are: $\mathcal{R}_v= 2.1, 1.7, 1.3$. The plot clearly demonstrates the critical boosting threshold at $\delta^* = s(1-\epsilon)$ (dashed vertical line), where ${\rm VE}$ exactly matches individual-level vaccine efficacy ($1-\epsilon$) independent of coverage. Notably, in the frequent boosting regime ($\delta > \delta^*$), the graph illustrates the ``${\rm VE}$ paradox", wherein increasing the vaccine coverage in the population paradoxically yields a lower relative vaccine effectiveness. {\bf C.} Overall vaccine effectiveness as a function of $\delta$, utilizing the exact same coverage levels and parameter values as in panel B.}
    \label{fig:MastercodeVE}
\end{figure}
  To gain an intuitive understanding of these results, it is helpful to first distinguish our perspective from prior work. We are primarily interested in the behavior of ${\rm VE}(\Lambda)$ as a function of the total cumulative exposure, $\Lambda$, over the whole epidemic. In contrast, the established literature has predominantly examined the monotonic decline of ${\rm VE}$ as the epidemic progresses, that is, as exposure, $\Lambda_t$, gradually accumulates over time, even while true vaccine efficacy, $1 - \epsilon$, remains constant~\cite{halloran1992interpretation, kahn2024examining,lipsitch2019depletion, lipsitch2021interpreting,nikas2023competing,smith1984assessment,tokars2020waning}. Because Expression \eqref{eq:VELambda} formulates ${\rm VE}$ fundamentally as a function of cumulative exposure, it provides a conceptual bridge between these two perspectives, allowing the intuition developed for progressing epidemics to inform our whole-epidemic approach.
  
 This phenomenon is driven by what the epidemiological literature describes as ``differential depletion of susceptibles''~\cite{kahn2024examining}, the ``leaky vaccine'' effect~\cite{tokars2020waning}, or the ``frailty effect'' (also referred to as ``survivor bias'')~\cite{nikas2023competing}. As an epidemic progresses, the susceptible population at risk reduces due to infections. These reductions are inherently greater in the highly susceptible pool than in the group that is less at risk, e.g., the vaccinated group. As a result, the unvaccinated pool is depleted more rapidly, causing its incidence rate to decline over time relative to that in the vaccinated group. Consequently, the relative benefit of the vaccine wanes with an increases in the cumulative force of exposure.

 Within our framework, the instantaneous, per capita rate at which unvaccinated individuals leave the susceptible pool equals $s\lambda(t)=(\delta^*+s\epsilon)\lambda(t)$, while the rate at which vaccinated individuals leave their susceptible compartment, via either infection or immune boosting, is $(\delta + s\epsilon)\lambda(t)$. The ratio of these depletion rates, \[\frac{\delta+s\epsilon}{\delta^*+s\epsilon},\] perfectly dictates the three distinct structural regimes established in Proposition \ref{prop:VEanalyticMonotone},  as illustrated in Figure \ref{fig:MastercodeVE}A:
 
\textbf{Case 1: $\delta < \delta^*$ (Weak immune boosting).} 
In this scenario, the unvaccinated cohort experiences a greater proportional depletion of its susceptible population than the vaccinated group as cumulative exposure increases. This creates a classic differential depletion of susceptibles bias~\cite{kahn2024examining,nikas2023competing,tokars2020waning}

\textbf{Case 2: $\delta > \delta^*$ (Strong Immune Boosting).} 
Under strong immune boosting, the depletion pattern inverts. Frequent exposures shuttle vaccinated individuals into a protected state without infection. Consequently, the vaccinated susceptible pool experiences a greater proportional depletion than the unvaccinated cohort across cumulative exposure.

\textbf{Case 3: $\delta = \delta^*$ (Balanced Dynamics).} 
In this case, the proportional depletions of the two susceptible pools are perfectly balanced across all levels of exposure, eliminating the asymmetry that drives the survivorship bias. The exponential terms in \eqref{eq:VELambda} cancel entirely, so that vaccine effectiveness remains constant regardless of cumulative exposure, equaling the baseline vaccine efficacy, ${\rm VE}(\Lambda) = 1 - \epsilon$  (see dashed red line in Figure \ref{fig:MastercodeVE}A).

\subsection{The Vaccine Effectiveness Paradox and Overall Impact}
Conventionally, higher vaccine coverage is expected to increase vaccine effectiveness by reducing population-level transmission pressure. This intuition is supported by classical frameworks~\cite{halloran1999design}, where increased vaccine coverage~$\phi$, results in higher vaccine effectiveness. However, in the presence of strong immune boosting ($\delta > \delta^*$), we now show that this relationship is reversed. 
We refer to this phenomenon, where broader population protection diminishes vaccine effectiveness, as the {\em Vaccine Effectiveness Paradox}.

Since in the case $\delta>\delta^*$, ${\rm VE}(\Lambda)$ is strictly increasing (Proposition~\ref{prop:VEanalyticMonotone}), the sign of ${\rm VE}'(\phi)$ is governed by $\Lambda'(\phi)$. To show that vaccine effectiveness declines as coverage increases, it is therefore sufficient to establish that $\Lambda'(\phi) < 0$, as illustrated by the arrow in Figure \ref{fig:MastercodeVE}A. 

The following proposition shows that for the case $\frac{\beta_v}{\gamma_v}\leq \frac{\beta_u}{\gamma_u}$, $\Lambda'(\phi)<0$ for all boosting values $\delta>0$.

\begin{restatable}{proposition}{vaccinationboosting} \label{prop:vaccination_boosting}
Suppose that $\mathcal{R}_v > 1$, and that condition \eqref{eq:vaccineassumption} holds. For all $\delta>0$, the cumulative force of exposure $\Lambda$, which is the solution of Equation \eqref{eq:finalsizelambda}, strictly decreases as vaccine coverage $\phi$ increases, i.e.,
\begin{equation}
    \frac{\partial \Lambda}{\partial \phi} < 0.
\end{equation}
\end{restatable}
\begin{proof}
    See Appendix \ref{sec:VEparadoxappendix}.
\end{proof}
Proposition \ref{prop:vaccination_boosting} formalizes the conditions under which increasing vaccine coverage $\phi$ reduces the cumulative force of exposure $\Lambda$ ($\Lambda'(\phi) < 0$). This suppression of transmission chains is the primary driver of the vaccine effectiveness paradox observed in the frequent boosting regime ($\delta > \delta^*$). Because ${\rm VE}$ is a strictly increasing function of $\Lambda$ in this regime (see Proposition \ref{prop:VEanalyticMonotone}), the success of a public health campaign in lowering disease pressure counterintuitively drives the observed ${\rm VE}$ downward, see Figure \ref{fig:MastercodeVE}A. The crossover behavior demonstrated in Figure \ref{fig:MastercodeVE}B highlights this dynamic: as coverage successfully protects the population, the measured effectiveness of the vaccine declines, prompting the question of how a more successful campaign can lead to an apparently lower vaccine performance.

Ultimately, the driving mechanism of the ``paradox" is the ``over" benefit of vaccination and related immune boosting events. By successfully lowering the disease pressure and reducing the number of pathogen encounters, the public health campaign limits beneficial immune-boosting events, which counterintuitively drives down the relative effectiveness of the vaccine.

It is important to note that while higher vaccine coverage paradoxically lowers this observed relative vaccine effectiveness when boosting is high, it still successfully reduces the total disease burden, that is overall infections across the population. To demonstrate this, and to show that the paradox is an artifact of the mathematical interplay between relative attack rates, we evaluate the population-level impact using the overall vaccine effectiveness, ${\rm VE}_{\rm overall}$ \cite{shim2012distinguishing}. This metric captures the total vaccination impact by comparing the final epidemic size, $Z$, to a baseline scenario where no vaccination is given:
\begin{equation}\label{eq:VEoverall}
    {\rm VE}_{\rm overall} = 1 - \frac{Z(\delta, \epsilon, s, \phi)}{Z(\delta, \epsilon, s, \phi=0)},
\end{equation}
where the total final size $Z$ is the weighted sum of the attack rates in each group:
\begin{equation}\label{eq:Z}
    Z = (1-\phi)AR_u + \phi AR_v.
\end{equation}
In Equation \eqref{eq:VEoverall}, we represent the absence of vaccination by setting the vaccine coverage parameter to $\phi=0$.

Notably, ${\rm VE}_{\rm overall}$ depends on the specific boosting probabilities, $q_s$ and $q_w$, exclusively through the aggregate parameter $\delta$. 
Using Proposition \ref{prop:vaccination_boosting}, we prove the following:
\begin{restatable}{proposition}{VEoevrall} \label{prop:VEoevrall}
Assuming condition \eqref{eq:vaccineassumption} holds, for all $\delta>0$, ${\rm VE}_{\rm overall}$ increases with the vaccine coverage, $\phi$.
\end{restatable}
\begin{proof}
    See Appendix \ref{sec:VEoverallApp}.
\end{proof}
Unlike ${\rm VE}$, ${\rm VE}_{\rm overall}$ cannot be expressed as a function of a single $\Lambda$, as the numerator and denominator satisfy distinct formulations of Equation \eqref{eq:FinalSizeLambda}. Therefore, the computation of ${\rm VE}_{\rm overall}$ requires solving the implicit final size equation~\eqref{eq:finalsizelambda}. Figure \ref{fig:MastercodeVE}C demonstrates the result of Proposition \ref{prop:VEoevrall}: In the case $\frac{\beta_v}{\gamma_v}\leq\frac{\beta_u}{\gamma_u}$, ${\rm VE}_{\rm overall}$ increases monotonically with vaccine coverage $\phi$, for all $\delta>0$. Critically, this metric exhibits no crossover behavior or critical thresholds. These results confirm that while higher coverage may paradoxically lower the observed vaccine effectiveness measured by~${\rm VE}$, it consistently reduces the total disease burden.

\subsubsection{A note on Net-Beneficial Vaccines}
We can expand the condition \eqref{eq:vaccineassumption} to consider a net-beneficial vaccine \eqref{eq:net_beneficial}, 
meaning that the vaccine's reduction of susceptibility compensates for any potential increase in baseline transmissibility. Following the proof of Proposition \ref{prop:vaccination_boosting}, for $\delta>\delta^*$, it still holds that $\Lambda'(\phi) < 0$ for a net-beneficial vaccine. However, the relationship between coverage and effectiveness is more varied in the infrequent boosting regime ($\delta < \delta^*$) in the case of a net-beneficial vaccine \eqref{eq:net_beneficial}. When $\frac{\beta_v}{\gamma_v} \le \frac{\beta_u}{\gamma_u}$, the aggregate disease pressure is successfully suppressed, and the observed ${\rm VE}$ increases with vaccine coverage $\phi$, as seen in Figure \ref{fig:MastercodeVE}B. However, under conditions where $\frac{\beta_v}{\gamma_v}>\frac{\beta_u}{\gamma_u}$, we observe instances where ${\rm VE}'(\phi)$ is negative or changes sign from negative to positive in the infrequent boosting regime, see the supplementary figures in Appendix~\ref{sec:appendixnumericalsimulations}. 
Thus, the intuitive epidemiological expectation that increased coverage necessarily reduces aggregate disease pressure ($\Lambda'(\phi) < 0$) fails to universally hold in the infrequent boosting regime.
Notably, ${\rm VE}_{\rm overall}$ does not exhibit this behavior. Rather, it consistently increases with $\phi$ for all $\delta$ in the case of a net-beneficial vaccine.

\subsection{Behavior of Highly Transmissible Pathogens}
In this section we focus on the limit of highly transmissible pathogens. This analysis is especially relevant given the emergence of highly transmissible variants like Omicron in COVID-19 \cite{burki2022omicron}.
We derive the attack rates for epidemics with large ${\mathcal R}_v$ by taking $\beta_u\to \infty$ or $\beta_v\to \infty$ in \eqref{eq:finalsizelambda} while holding all other parameters constant, and substituting the result into \eqref{eq:finalsizeULambda} and \eqref{eq:finalsizeVLambda}. 
\begin{restatable}{proposition}{ARasymptoticLambda}\label{prop:ARasymptotitcLambda}
 Let $\phi>0$. The attack rates satisfy:
\begin{equation}\label{eq:ARlargeRv}
    \begin{split}
     &\lim_{\beta_u\to \infty} AR_u=\lim_{\beta_v\to \infty} AR_u=\lim_{\Lambda\to \infty} AR_u=1,\\
     &\lim_{\beta_u\to \infty} AR_v=\lim_{\beta_v\to \infty} AR_v=\lim_{\Lambda\to \infty} AR_v=\frac{s\epsilon}{\delta+s\epsilon}.
    \end{split}  
\end{equation}
\end{restatable}

\begin{proof}
    See Appendix \ref{sec:highRvLambdaApp}. Note that we assume the vaccine is at least net-beneficial \eqref{eq:net_beneficial}. Therefore, if $\beta_v \to \infty$, then $\beta_u \to \infty$.
\end{proof}
For highly transmissible pathogens, nearly all unvaccinated individuals are infected.
Furthermore, in a leaky vaccination SIR model, almost all vaccinated individuals will eventually be infected as well \cite{park2023immune}.
 However, when immune boosting is present, i.e. $\delta > 0$, Equation \eqref{eq:ARlargeRv} demonstrates that not all vaccinated individuals become infected, even during widespread outbreaks, as some of them gain immunity via boosting. Ultimately, only a constant fraction of the vaccinated population, determined solely by boosting dynamics and vaccine efficacy, becomes infected. 

Substituting the limiting attack rates \eqref{eq:ARlargeRv} into the definition of ${\rm VE}_{\rm overall}$, we get:
\begin{equation}\label{eq:VEoverallLArgerv}
\lim_{\beta_u\to\infty} {\rm VE}_{\rm overall}=\lim_{\beta_v\to\infty} {\rm VE}_{\rm overall}=\phi\left(1-\frac{s\epsilon}{\delta+s\epsilon}\right).
 \end{equation}
Particularly, the overall vaccine effectiveness remains proportional to coverage under extreme transmission rates and boosting.

Substituting \eqref{eq:ARlargeRv} into the definition of ${\rm VE}$, we obtain the asymptotic vaccine effectiveness:
\begin{equation}\label{eq:largeRv}
     \lim_{\beta_u\to\infty} \rm{VE}=\lim_{\beta_v\to\infty} \rm{VE}=1-\frac{s\epsilon}{\delta+s\epsilon}.
\end{equation}

 Equation~\eqref{eq:largeRv} formalizes a fundamental distinction in how population-level vaccine effectiveness (VE) behaves under intense transmission. In the absence of boosting, a massive epidemic driven by a highly transmissible pathogen inevitably infects virtually the entire population regardless of vaccination status, driving the VE to zero. In contrast, the boosting mechanism confers complete protection to a subset of individuals, allowing VE to logically escape this collapse and converge to a stable, positive limit. Yet, while this non-zero limit is mathematically expected, Equation~\eqref{eq:largeRv} reveals a genuinely surprising property: this robust relative advantage is structurally invariant, remaining entirely independent of the overall vaccine coverage,$\phi$.

 Although the detailed ODE model captures this behavior, its intertwined feedback loops obscure the precise mechanism driving it. The full framework includes the temporal accumulation of repeated pathogen exposures, the time-dependent slowdown of the epidemic as susceptibles are depleted and the structural variation in~$\mathcal{R}_v$ across different levels of vaccine coverage.
In Section \ref{sec:approximateVE}, we strip the model of these overlapping factors to isolate the core driver. By transitioning to a individual-level probabilistic framework, we demonstrate that the independence of relative vaccine effectiveness from vaccine coverage for highly contagious pathogens is ultimately governed by a single, fundamental mechanism: repeated exposures.

\section{Vaccine Effectiveness Under Repeated Pathogen Exposures}\label{sec:approximateVE}
In this section, we introduce an individual-level probabilistic framework to isolate the direct impact of repeated pathogen exposures. Rather than tracking population-wide transmission dynamics, this framework provides an exact calculation of an individual's probability of infection, conditional on experiencing a specific number of exposures. By focusing on these precise individual trajectories, we clarify the fundamental mechanisms driving the counterintuitive behaviors observed in Section~\ref{sec:VE}. In particular, this exact conditional analysis exposes the structural origin of the critical boosting threshold~$\delta^*$, provides a clear mechanistic basis for the vaccine effectiveness paradox, and demonstrates why relative protection decouples from vaccine coverage under high transmission limits.

The total number of exposures $N$ that an individual experiences during the epidemic is fundamentally determined by the cumulative force of exposure, $\Lambda$, in the full model~\eqref{eq:Model}. Because the instantaneous rate of exposure is given by $\lambda(t)$, the actual number of exposures $N$ for an individual is a random variable following a Poisson distribution. The mean of this distribution is exactly the cumulative force of exposure \eqref{eq:Lambdadef}.
Consequently, the probability that an individual undergoes exactly $n$ exposures is given by
\begin{equation}\label{eq:probability_of_N}
    P(N=n) = e^{-\Lambda} \frac{\Lambda^{n}}{n!}.
\end{equation}

In what follows, we evaluate protection from the perspective of a single individual by conditioning on a fixed number of discrete pathogen exposures, $n$. This conditional framework completely decouples individual outcomes from population-wide transmission dynamics and time-varying risks. While the individual's total probability of infection over the course of the epidemic is recovered by compounding these conditional probabilities with the Poisson distribution of exposures~\eqref{eq:probability_of_N}, analyzing a fixed $n$ isolates the pure, cumulative effect of repeated encounters on a single person. This method of tracking individual exposure histories builds on the framework introduced by \cite{halloran1991direct}, which we extend here to incorporate immune boosting. Because this analysis represents a single individual rather than a population, it cannot capture the group-level shifts driven by vaccination coverage and is therefore unsuitable for evaluating overall vaccine effectiveness, ${\rm VE}_{\rm overall}$.

\subsection{Evaluating ${\rm VE}_n$}
 Define ${\rm VE}_n$ as the vaccine effectiveness given $n$ exposures. It is given by: 
\begin{equation}\label{eq:VEn}
  \rm{VE}_n= 1-\frac{Q_n^{v}}{Q_n^{u}}, 
\end{equation}
 where $Q_n^{u}$ denotes the probability of infection for an unvaccinated individual following $n$ exposures, and $Q_n^{v}$ represents the corresponding probability for a vaccinated individual.

To compute $Q_n^{u}$, we analyze the possible outcomes for a representative unvaccinated individual across multiple disease exposures. Following exposure, an unvaccinated individual becomes infected, $I_u$, with probability $s$, or remains susceptible with the complementary probability $1-s$, see Figure \ref{fig:diagram}. Thus, $Q_{1}^{u}=s$. Since $Q_n^{u}$ is the complement of the survival probability over $n$ exposures to the pathogen, we have:
\begin{equation}\label{eq:Qnudirect}
    Q_n^{u}=1-(1-s)^{n}.
\end{equation}

Similarly, to compute $Q_n^{v}$, we examine  the possible outcomes for representative vaccinated individual across multiple pathogen exposures, see Figure \ref{fig:diagram}. A vaccinated individual transitions to the infected class, $I_v$, if they experience a regular exposure, occurring with probability $s$, that successfully leads to infection, with probability $\epsilon$, resulting in a joint transition probability of $s\epsilon$.
Alternatively, vaccinated individuals may transition to the boosted compartment, $R_b$, where they acquire complete protection against future infections. This boosting process occurs via two distinct pathways: Following a regular exposure that fails to result in infection, with probability $sq_s(1 - \epsilon)$, or following a weak exposure, with probability $q_w(1 - s)$. 
Thus, the resulting probability of remaining susceptible following exposure for vaccinated individuals is $1-\delta-s\epsilon$, as defined in Equation \eqref{eq:delta}.

\begin{figure}[ht!]
    \centering
\includegraphics[width=1\linewidth]{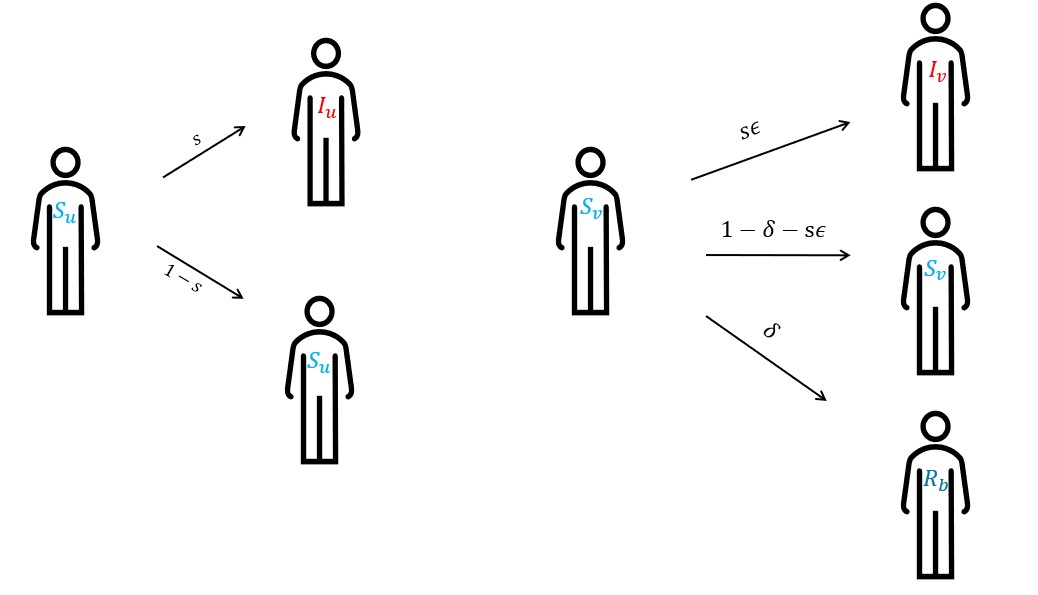}
     \caption[Diagram depicting the possible outcomes following exposure to infection in vaccinated and unvaccinated individuals]
    {Diagram depicting the possible outcomes following exposure to infection in vaccinated and unvaccinated individuals, including infection, unchanged susceptibility, and immune boosting. 
    } 
    \label{fig:diagram}
\end{figure}

The probability of infection for a vaccinated susceptible after one exposure is $\left.Q_{1}^{v}=s\epsilon\right.$.
The probability of infection after $n$ exposures is equal to one minus the probability of remaining susceptible throughout all $n$ exposures without being boosted, and minus the probability of transitioning to a boosted state at the $i$-th exposure , $0\leq i < n$:
\begin{equation}\label{eq:Qnvdirect}
    Q_n^{v}=1-(1-\delta-s\epsilon)^{n}-
    \delta\sum_{i=0}^{n-1} (1-\delta-s\epsilon)^{i}=s\epsilon\left(\frac{1-(1-\delta-s\epsilon)^{n}}{\delta+s\epsilon}\right).
\end{equation}
Substituting \eqref{eq:Qnudirect} and \eqref{eq:Qnvdirect} into \eqref{eq:VEn}, we get that the vaccine effectiveness given $n$ exposures to the pathogen is given by:
\begin{equation}\label{eq:VEncalc}
    {\rm VE}_n=1-\frac{[1-(1-\delta-s\epsilon)^{n}]s\epsilon}{[1-(1-s)^{n}](\delta+s\epsilon)},\quad n\geq 1.
\end{equation}
Note that this expression for vaccine effectiveness following $n$ exposures closely resembles the formula \eqref{eq:VELambda} for real world vaccine effectiveness derived from the full model.

\begin{remark}
It is possible to relate $Q_n^u$ and $Q_n^v$ to $AR_u$ and~$AR_v$, respectively.
The probability of a vaccinated individual being infected throughout the epidemic is:
\begin{equation}\label{eq:connection}
    \sum_{n=0}^{\infty} Q_n^{v} P(N=n)=\frac{s\epsilon}{\delta+s\epsilon}\sum_{n=0}^{\infty} \left(1-(1-\delta-s\epsilon)^{n}\right)e^{-\Lambda}\cdot \frac{\Lambda^{n}}{n!},
\end{equation}
where~$Q_n^v$ is given by~\eqref{eq:Qnvdirect}, and~$P(N=n)$ is given by~\eqref{eq:probability_of_N}, 
Therefore,
\begin{equation*}\begin{split}
    AR_v&= \frac{s\epsilon}{\delta+s\epsilon} e^{-\Lambda} \left[\sum_{n=0}^{\infty}\frac{\Lambda^{n}}{n!} -\sum_{n=0}^{\infty}(1-\delta-s\epsilon)^{n}\cdot \frac{\Lambda^{n}}{n!}\right]\\
    &= \frac{s\epsilon}{\delta+s\epsilon} e^{-\Lambda} \left[ e^{\Lambda}-e^{(1-\delta-s\epsilon)\Lambda}\right]=\frac{s\epsilon}{\delta+s\epsilon}[1-e^{-(\delta+s\epsilon)\Lambda}].
    \end{split}
\end{equation*}
This is exactly Equation \eqref{eq:finalsizeVLambda}. A similar computation can be carried out to relate $Q_n^u$ to $AR_u$.
While this calculation yields Equations  \eqref{eq:finalsizeULambda} and \eqref{eq:finalsizeVLambda}, $\Lambda$ (Equation \eqref{eq:finalsizelambda}) remains undetermined.  Solving for $\Lambda$ requires additional data regarding $\mathcal{R}_v$ and other baseline parameters.
\end{remark}
\subsection{Behavior of 
$\rm{VE}_n$}

We first characterize the behavior of ${\rm VE}_n$ as a function of $\delta$, holding $s$ and $\epsilon$ constant while changing $q_s,q_w$.
\begin{restatable}{proposition}{VEmmonotnedelta}\label{prop:VEmmonotnedelta}
    Let $n>1$. For any $\delta_1>\delta_2$, ${\rm VE}_n(\delta_1)>{\rm VE}_n(\delta_2)$.
\end{restatable}
\begin{proof}
    See appendix \ref{sec:VEintuitionAppendics}.
\end{proof}

This result is illustrated in Figure \ref{fig:differentnsfunc ofdelta}. Note that for $n=1$, ${\rm VE}(\delta)=1-\epsilon$ for all $\delta$.
\begin{figure}[H]
    \centering
\includegraphics[width=0.8\linewidth]{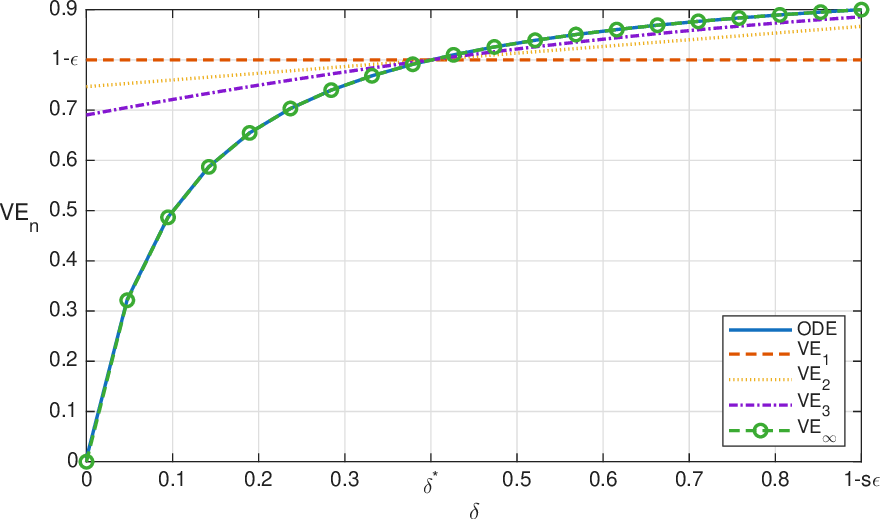}
    \caption[$\rm{VE}_n$  as a function of $\delta$, calculated for $n=1$, $n=2$, $n=3$, $n\to\infty$, and ${\rm VE}$ for high reproduction number]
    {$\rm{VE}_n$  as a function of $\delta$, calculated for $n=1$ (dashed red), $n=2$ (dotted yellow), $n=3$ (dash-dotted purple) where $\epsilon=0.2$, and $s=0.5$. The plot demonstrates that ${\rm VE}_n$ is monotone increasing in $\delta$. For $\delta<\delta^{*}$, ${\rm VE}_n(\delta)$ is monotone decreasing in $n$. For $\delta=\delta^{*}$, ${\rm VE}_n(\delta)$ is equal to ${\rm VE}_1=1-\epsilon$, and for $\delta>\delta^{*}$, ${\rm VE}_n(\delta)$ is monotone increasing in $n$.
    The green dashed curve with open circle markers represents ${\rm VE}_\infty$ \eqref{eq:VEnisinf}. The blue solid curve represents the ODE based calculation of ${\rm VE}$ for high reproduction number (${\mathcal R}_v\approx13$), as discussed in Section \ref{sec:VE}.} 
   \label{fig:differentnsfunc ofdelta}
\end{figure}
Proposition \ref{prop:VEanalyticMonotone} outlines the behavior of ${\rm VE}$ as a function of $\Lambda$ under various regimes of the boosting parameter $\delta$. The following proposition applies this framework to ${\rm VE}_n$, establishing parallel monotonicity: ${\rm VE}_n$ decreases monotonically with the number of exposures $n$ for $\delta<\delta^{*}$, and increases monotonically for $\delta>\delta^{*}$. 
We prove the following:
\begin{restatable}{proposition}{VEmmonotneinn}\label{prop:VEmmonotneinn}
    Let $n>1$. For $\delta<\delta^{*}$, we have ${\rm VE}_{n+1}<{\rm VE}_n$, and for $\delta>\delta^{*}$, ${\rm VE}_{n+1}>{\rm VE}_n$. In the case $\delta=\delta^{*}$, we get ${\rm VE}_n={\rm VE}_1$. 
\end{restatable}
\begin{proof}
    See appendix \ref{sec:VEintuitionAppendics}.
\end{proof}
This individual-level probabilistic framework perfectly recovers the critical threshold, $\delta^* =
(1-\epsilon)s$, that originally emerged from the complex continuous-time ODEs in Proposition \ref{prop:VEanalyticMonotone}. Similar to ${\rm VE}$, for $\delta=\delta^*$, ${\rm VE}_n$ is also constant and equal to the individual level vaccine efficacy regardless of the number of exposures. 

This result makes the dynamics of the threshold mathematically transparent, noting that $\delta=\delta^{*}=(1-\epsilon)s$ corresponds to the case where $q_s=1$ and $q_w=0$. In this scenario, subsequent exposures have no effect: upon their initial exposure, an individual either becomes infected or acquires permanent immunity. Therefore, ${\rm VE}_n={\rm VE}_1$ remains constant. The same effective behavior holds for all systems with $\delta=\delta^{*}$.

This individual-level analysis can suggest an intuitive explanation to the vaccine effectiveness paradox. In the case of net-beneficial vaccine, higher vaccine coverage reduces the reproductive number, see \eqref{eq:net_beneficial}, which in turn means individuals experience fewer pathogen exposures (smaller $n$). Therefore, under strong boosting conditions, $\delta > \delta^*$, where effectiveness increases with more exposures, a higher vaccine coverage, which reduces exposures, paradoxically leads to a lower ${\rm VE}$. 
\subsection{Asymptotic Behavior of ${\rm VE}_n$}
In the limit $n \to \infty$, the probabilities of infection approach
\begin{equation}
    \lim_{n\to\infty}Q_{n}^{u}=1,\qquad 
    \lim_{n\to\infty}Q_{n}^{v}=\frac{s\epsilon}{\delta+s\epsilon},
\end{equation}
see~\eqref{eq:Qnudirect} and~\eqref{eq:Qnvdirect}.  
Consequently, the asymptotic vaccine effectiveness is:
\begin{equation}\label{eq:VEnisinf}
    {\rm VE}_{\infty}=\lim_{n\to\infty}{\rm VE}_n=1-\frac{s\epsilon}{\delta+s\epsilon}.
\end{equation}
These results recover the asymptotic attack rates,~\eqref{eq:ARlargeRv}, and the asymptotic~${\rm VE}$, ~\eqref{eq:largeRv}, arising from the corresponding ODE-based formulation.
Figure \ref{fig:differentnsfunc ofdelta} presents ${\rm VE}_n(\delta)$ for various values of $n$ and illustrates these results.

Furthermore, because ${\rm VE}_n$ is monotonic with respect to the number of exposures $n$, we can establish strict bounds on vaccine effectiveness,
\begin{equation}\label{eq:VEnbounds}
\begin{split}
    & \forall \delta<\delta^{*}: {\rm VE}_{\infty} \leq{\rm VE}_n\leq {\rm VE}_1,\\
    & \forall\delta>\delta^{*}: {\rm VE}_{1} \leq{\rm VE}_n\leq {\rm VE}_\infty.
\end{split}
\end{equation}
These bounds align with those established analytically for population-level ${\rm VE}$, in Proposition \ref{prop:VEanalyticMonotone}.

Shifting to this individual-level probabilistic framework yields three core insights into the counterintuitive behavior of vaccine effectiveness (${\rm VE}$): First, it explains the ${\rm VE}$ paradox. Broader vaccine coverage successfully reduces the overall frequency of pathogen exposures across the population. However, in a frequent boosting regime ($\delta > \delta^*$), this reduction in exposure frequency paradoxically drives down the relative effectiveness of the vaccine.
Second, the model perfectly recovers the critical threshold at $\delta = \delta^*$, demonstrating that when exposure leads to either infection or permanent immunity, ${\rm VE}$ remains static and equal to individual-level vaccine efficacy ($1-\epsilon$).

The individual-level probabilistic framework recovers the asymptotic behavior of ${\rm VE}$  for highly transmissible pathogens. At high reproductive numbers, ${\rm VE}$  converges to a limit corresponding to infinitely many exposures, which is inherently independent of vaccine coverage.
To understand why the framework successfully predicts asymptotic outcomes at high reproductive numbers, we must evaluate the three competing effects captured by the full ODE model:
\begin{enumerate}[label=\Roman*., align=left, labelwidth=2em, labelsep=0.5em, leftmargin=3em]
\item \textbf{Multiple exposures.}
\item \textbf{Reduced transmission:} $\mathcal{R}_v$ drops as vaccine coverage increases.
\item \textbf{Epidemic slowdown:} Susceptibles are depleted as recovered individuals accumulate.
\end{enumerate}

Our findings confirm that under high transmission, multiple pathogen exposures completely dominate the other two effects. While broader vaccine coverage does reduce transmission, the reproduction number remains large enough that the asymptotic limits still hold. Furthermore, as $\mathcal{R}_v \to \infty$, the accumulation of recovered individuals happens so quickly that nearly everyone is infected or boosted instantaneously, so that epidemic slowdown has no dynamic effect. Because the probabilistic framework exclusively isolates this dominant exposure mechanism, it sidesteps coverage-dependent transmission dynamics entirely, successfully explaining why ${\rm VE}$ ultimately decouples from vaccine coverage at the high transmission limit.

\section{Discussion}\label{sec:discussion}

In this study, we established a mathematical framework to evaluate the population-level impact of immune boosting during a single infectious disease outbreak \cite{lavine2011natural, park2023immune}. Ultimately, our analysis provides a fundamental approach to better understand the complex, emergent behavior of population-level vaccine effectiveness, ${\rm VE}$. By extending the classical SIR model to incorporate a mechanism of perfect boosting, triggered by both regular and weak pathogen exposures, we derived closed-form final size relations that explicitly connect attack rates to the boosting parameter $\delta$. Central to our findings is the identification of a critical boosting threshold, $\delta^* = (1-\epsilon)s$, which strictly governs the system's dynamics. Below this threshold, relative ${\rm VE}$ behaves conventionally. At the threshold, ${\rm VE}$ exactly matches individual-level vaccine efficacy, and above the threshold, we observe a regime where relative ${\rm VE}$ increases with the force of exposure but paradoxically decreases with vaccine coverage.

Our theoretical findings offer essential context for the broader epidemiological literature dedicated to estimating vaccine effectiveness and individual vaccine efficacy \cite{halloran1992interpretation,halloran1991direct, halloran1999design, shim2012distinguishing,smith1984assessment}. In classical mathematical frameworks, vaccine effectiveness is typically expected to decline as the total scale of the outbreak increases, driven by the differential depletion of susceptible, a mechanism often described as the leaky vaccine effect or survivor bias \cite{kahn2024examining,lipsitch2019depletion,  nikas2023competing,tokars2020waning}. Indeed, our model analytically recovers this expected monotonic decline in the infrequent boosting regime ($\delta < \delta^*$). However, we demonstrate that under frequent immune boosting ($\delta > \delta^*$), the depletion pattern inverts, producing a counterintuitive trajectory where relative vaccine effectiveness increases with the cumulative force of exposure. Consequently, our findings expand traditional frameworks by showing that relative protection is a context-dependent metric shaped by the accumulation of silent pathogen exposures. Therefore, in the presence of immune boosting, the standard statistical techniques used to infer vaccine efficacy from observational data warrant careful reexamination.

Our findings suggest that public health surveillance capable of tracking immune boosting, such as through serial serological surveys and contact tracing, could help clarify the interpretation of vaccine effectiveness metrics in the presence of immune boosting. Without such adjustments, there is a potential for data misinterpretation, which we characterize as the vaccine effectiveness paradox: a population with higher vaccine coverage can paradoxically exhibit lower relative vaccine effectiveness than a comparable population with lower coverage. Consequently, a lower relative protection metric in a highly vaccinated community might be falsely attributed to vaccine failure or variant escape, rather than the epidemiological success of reducing the natural ``boosters" required to maintain a high relative advantage. To avoid this pitfall, evaluating population-level impact requires moving away from a strict reliance on relative individual risk measurements. The true success of a vaccination campaign is captured by the overall vaccine effectiveness, ${\rm VE}_{\rm overall}$, which measures the absolute reduction in the epidemic burden. Because ${\rm VE}_{\rm overall}$ is a counterfactual metric that cannot be directly observed within a single population, we point to several potential avenues for future empirical research that could help protect observational studies from this bias. First, future investigations might explore complementing traditional relative metrics (such as hazard ratios) with absolute measures of disease burden, including incidence rate differences and the absolute risk reductions between cohorts. Second, public health frameworks could evaluate the feasibility of cross-jurisdictional or multi-center cohort designs to compare absolute incidence rates across distinct regions with varying levels of vaccine coverage. Finally, rather than relying on binary seropositivity assays, future surveillance programs might benefit from exploring quantitative serological tracking, monitoring shifts in pre- and post-outbreak antibody titer distributions, to infer the background rate of subclinical, immunity-reinforcing exposure events across different coverage regimes.

Methodologically, our approach underscores the advantages of utilizing the cumulative force of exposure, $\Lambda$, as the natural intrinsic coordinate of an epidemic. In particular, formulating the system with respect to $\Lambda$ substantially simplifies the derivation of final size equations. By shifting analytical focus to the cumulative force of exposure, we were able to analyze vaccine effectiveness while bypassing the complexities involved in the solution of Equation~\eqref{eq:finalsizelambda}.

While our model provides a foundational understanding of perfect boosting during short-scale outbreaks, it inherently relies on specific assumptions that suggest avenues for future extensions. We assumed that boosting confers full and permanent immunity, yet biological reality often involves partial boosting, yielding higher but incomplete protection against infection upon subsequent exposure \cite{barbarossa2015mathematical}. Furthermore, expanding the framework to account for asymptomatic individuals acting as ``semi-boosters" with reduced infectivity could provide a more nuanced picture of disease transmission. It would also be valuable to explore a model in which unvaccinated individuals can gain immune protection following exposure which does not lead to infection.
Finally, while appropriate for single outbreaks, adapting this model to long-term endemic scenarios will require the re-introduction of waning immunity and demographic turnover, allowing for a comprehensive analysis of multi-wave epidemic trajectories. Looking forward, such extended frameworks will be valuable for evaluating next generation technologies, such as endogenous `self-boosting' vaccines \cite{arinaminpathy2012self}, and devising an optimized public health strategy to suppress overall disease burden over prolonged periods of time.

\bmhead{Acknowledgements}

This research was supported by the Israel Science Foundation (grant no. 3730/20) within the KillCorona-Curbing Coronavirus Research Program, and by the Israel Science Foundation (ISF) grant 1596/23.

\begin{appendices}
\section{Derivation of Post-Vaccination Basic Reproduction Number}\label{secA:Rv}
The model describes cases for which vaccination is administrated before the epidemic starts. To compute the basic reproduction number after vaccination, $\mathcal{R}_v$, we follow
the standard computation of the next generation matrix \cite{martcheva2015introduction}. We arrange the equations for the infected compartments $I_u, I_v$ in the following way:  
\begin{equation}\label{eq:Icompartments}
\begin{split}
     &I_u^\prime=\underbrace{s\lambda(t)S_u}_{\mathcal{F}_1}-\underbrace{\gamma_uI_u}_{\mathcal{V}_1},\\
&I_v^\prime=\underbrace{s\epsilon\lambda(t)S_v}_{\mathcal{F}_2}-\underbrace{\gamma_vI_v}_{\mathcal{V}_2},\\
\end{split}
\end{equation}
where $\mathcal{F}_i$ is the appearance rate of new infections in the vaccinated and unvaccinated compartments, and $\mathcal{V}_i$
incorporates
the remaining transitional terms: disease progression and recovery.
We linearize the system of infected compartments, (\ref{eq:Icompartments}), about
the disease-free equilibrium where the entire  population is susceptible, and obtain:
\begin{equation}
    x^\prime=(F-V)x,\quad x=(I_u, I_v)^T,
\end{equation}
where
\begin{align}
    F=\begin{pmatrix}
        s\beta_u(1-\phi)& s\beta_v(1-\phi)\\ s\epsilon \beta_u  \phi& s\epsilon\beta_v  \phi
    \end{pmatrix},\quad  V=\begin{pmatrix}
        \gamma_u& 0\\ 0 & \gamma_v
    \end{pmatrix}.
\end{align}

The post vaccination basic reproduction number is the spectral radius of the next-generation matrix:
$\mathcal{R}_v=\rho (FV^{-1}_{t=0})$, yielding \eqref{eq:Rv}.

\section{Proof of Proposition \ref{prop:finalsize}}\label{sec:FinalsizeappLambda}

Integrating equations \eqref{eq:ru}, \eqref{eq:rv} and substituting in definition \eqref{eq:ARdef}, we obtain:
\begin{equation}\label{eq:ARfromdef}
    AR_u= \frac{\gamma_u \int_{0}^{\infty} {I_u(t)\,dt}}{S_u(0)},\quad AR_v=\frac{\gamma_v \int_{0}^{\infty} {I_v(t)\,dt}}{S_v(0)}.
\end{equation}
From~\eqref{eq:Su} we get: 
\begin{equation}\label{eq:Su(infty)with Lambda}
S_u(\infty)=S_u(0)\exp{\left(-s\int_{0}^{\infty} \lambda(t)\,dt\right)} = S_u(0)\exp{(-s\Lambda)}.
\end{equation}
From~\eqref{eq:Su} and~\eqref{eq:Iu} we get: 
\[S_u^\prime(t)+I_u^\prime(t)=-\gamma_uI_u(t).\]
Integration yields:
\begin{equation*}
    S_u(\infty)-S_u(0)+\bcancel{I_u(\infty)}-\bcancel{I_u(0)}=-\gamma_u\int_{0}^{\infty} \ \hspace{-1em} I_u(t)\,dt,
\end{equation*}
so that at the limit $I_u(0)\to 0$ and since $I_u(\infty)=0$,
\begin{equation}\label{eq:gammauintIu}
    \gamma_u\int_{0}^{\infty}\ \hspace{-1em}I_u(t)\,dt=S_u(0)-S_u(\infty).
\end{equation}
Therefore, by substituting \eqref{eq:gammauintIu} and \eqref{eq:Su(infty)with Lambda} into \eqref{eq:ARfromdef}, the attack rate in the unvaccinated population is obtained by \eqref{eq:finalsizeULambda}.

Similarly, to calculate $AR_v$, from \eqref{eq:Sv}:
\begin{equation}\label{eq:Sv(infty)withLambda}
     S_v(\infty)=S_v(0)\exp(-(\delta+s\epsilon) \Lambda).
\end{equation}
From \eqref{eq:Sv} and \eqref{eq:Iv} we obtain
\[(\delta+s\epsilon) I_v^\prime+s\epsilon S_v^\prime=-(\delta+s\epsilon)\gamma_vI_v,  \]
so that at the limit $I_v(0)\to 0$ and since $I_v(\infty)=0$, integration yields:
\begin{equation}\label{eq:intgammavIv}
   \gamma_v\int_{0}^{\infty} \ \hspace{-1em} I_v(t)\,dt=\frac{s\epsilon}{\delta+s\epsilon}(S_v(0)-S_v(\infty)). 
\end{equation}
Substituting \eqref{eq:intgammavIv} and \eqref{eq:Sv(infty)withLambda} into~\eqref{eq:ARfromdef} yields \eqref{eq:finalsizeVLambda}.

Now, by substituting \eqref{eq:gammauintIu} and \eqref{eq:intgammavIv} into definitions \eqref{eq:Lambdadef} and \eqref{eq:lambda}, we obtain equation \eqref{eq:finalsizelambda} for $\Lambda$.

\section{Proof of Proposition \ref{prop:existenceLambda}}
\label{sec:UninqueLamdaapp}
Define $F(\Lambda)$ as:
\begin{equation}\label{eq:defoff}
    F(\Lambda)= \frac{\beta_u(1-\phi)}{\gamma_u}(1-\exp{(-s\Lambda)})+\frac{\beta_v\phi s\epsilon}{(\delta+s\epsilon)\gamma_v}\left(1-\exp{(-(\delta+s\epsilon)\Lambda)}\right)-\Lambda.
\end{equation}
Finding a solution to Equation \eqref{eq:finalsizelambda} is equivalent to finding a root of the function $F(\Lambda)$.
We aim to prove that $F(\Lambda) = 0$ has a unique solution $\Lambda^* > 0$ if and only if $\mathcal{R}_v > 1$.
First, observe that $F(0) = 0$. This confirms that $\Lambda = 0$ is always a solution, corresponding to the disease-free equilibrium of the system.

Next, we prove that $F(\Lambda)$ is strictly concave. The first derivative of $F(\Lambda)$ with respect to $\Lambda$ is:
\begin{equation*}
    F'(\Lambda) = \frac{\beta_u s(1-\phi)}{\gamma_u} e^{-s\Lambda} + \frac{\beta_v \phi s \epsilon}{\gamma_v} e^{-(\delta+s\epsilon)\Lambda}-1.
\end{equation*}

The second derivative of $F(\Lambda)$ is:
\begin{equation*}
    F''(\Lambda) = -\frac{\beta_u s^2(1-\phi)}{\gamma_u} e^{-s\Lambda} - \frac{\beta_v \phi s \epsilon (\delta+s\epsilon)}{\gamma_v} e^{-(\delta+s\epsilon)\Lambda}.
\end{equation*}
Given the positivity of the parameters, $F''(\Lambda) < 0$ for all $\Lambda \geq 0$. Therefore, $F(\Lambda)$ is strictly concave.

At the origin, the slope of $F$ is given by:
\begin{equation*}
    F'(0) = \mathcal{R}_v - 1.
\end{equation*} 
If $\mathcal{R}_v\leq1$, $F'(0) \leq 0$. Since $F(\Lambda)$ is strictly concave, the slope $F'(\Lambda)$ is strictly negative for all $\Lambda > 0$. Therefore, $F(\Lambda)$ strictly decreases away from $0$ and never intersects the horizontal axis again. The trivial solution $\Lambda = 0$ is the unique non-negative solution.

If $R_v > 1$, then $F'(0) > 0$. This implies that $F(\Lambda)$ initially increases, and $F(\Lambda) > 0$ for small $\Lambda > 0$. Notice that  $\lim_{\Lambda \to \infty} F(\Lambda) = -\infty$. Therefore, by the Intermediate Value Theorem, $F(\Lambda)$ has at least one positive root. Finally, the strict concavity of $F(\Lambda)$ ensures it can intersect the horizontal axis at most once for $\Lambda > 0$, guaranteeing that the positive solution $\Lambda^*$ is unique.

\section{Proof of Proposition \ref{prop:VEanalyticMonotone}}\label{sec:VEproofAPP}
For the case where $\delta =\delta^*$, substituting directly into Equation \eqref{eq:VELambda} yields ${\rm VE}(\Lambda)=1-\epsilon$ for all $\Lambda>0$.

For the case $\delta\neq \delta^*$, let 
\begin{equation*}
    R(\Lambda) = \frac{1 - \exp(-(\delta+s\epsilon)\Lambda)}{1 - \exp(-s\Lambda)}.
\end{equation*}
Since ${\rm VE}(\Lambda) = 1 - C \cdot R(\Lambda)$, where $C = \frac{s\epsilon}{\delta+s\epsilon}$ is a positive constant, the monotonicity of ${\rm VE}$ is exactly the opposite of the monotonicity of $R(\Lambda)$. Differentiating $R(\Lambda)$ with respect to $\Lambda$ gives:
\begin{equation*}
    R^\prime(\Lambda) = \frac{\exp(-(\delta+s\epsilon+s)\Lambda) \left[ (\delta+s\epsilon) e^{s\Lambda} - (\delta+s\epsilon) - s \exp((\delta+s\epsilon)\Lambda) + s \right]}{(1 - e^{-s\Lambda})^2}.
\end{equation*}

The denominator is always positive for $\Lambda > 0$, therefore the sign of $R'(\Lambda)$ is determined by the sign of the function:$$g(\Lambda) = (\delta+s\epsilon)(\exp(s\Lambda) - 1) - (\exp((\delta+s\epsilon)\Lambda) - 1)s.$$ Note that $g(0) = 0$. Differentiating $g(\Lambda)$ with respect to $\Lambda$ yields:
\begin{equation*}
    g'(\Lambda) = (\delta+s\epsilon)s \left[\exp((\delta^*+s\epsilon)\Lambda) - \exp((\delta+s\epsilon)\Lambda) \right].
\end{equation*}
If $\delta < \delta^*$,  then $g'(\Lambda) > 0$ for all $\Lambda > 0$,  which implies $R'(\Lambda) > 0$.
 However, if $\delta >  \delta^*$, then $g'(\Lambda) < 0$, meaning $R'(\Lambda) < 0$.
 Therefore, we have established that ${\rm VE}(\Lambda)$ is strictly decreasing when $\delta <  \delta^*$, and strictly increasing when $\delta >  \delta^*$. 
 
 Finally, to determine the bounds, we evaluate the limit of ${\rm VE}(\Lambda)$ as $\Lambda \to 0$. Applying L'Hôpital's rule, we obtain:
 \begin{equation*}
     \lim_{\Lambda\to 0} {\rm VE}= 1-\epsilon.
 \end{equation*}
Since ${\rm VE}(\Lambda)$ is strictly decreasing for $\delta < \delta^*$, it follows that ${\rm VE}(\Lambda)<1-\epsilon$ for all $\Lambda > 0$. Similarly, for $\delta > \delta^*$, ${\rm VE}(\Lambda)$ is strictly increasing, yielding ${\rm VE}(\Lambda)>1-\epsilon$. This completes the proof.

\section{Proof of Proposition \ref{prop:vaccination_boosting}}\label{sec:VEparadoxappendix}
By Proposition \ref{prop:existenceLambda}, $\Lambda^*$ is the unique positive root of the strictly concave function $F(\Lambda, \phi) = 0$. 

We apply the Implicit Function Theorem to $F$:
\begin{equation} \label{eq:implicit_F}
    \frac{\partial \Lambda^*}{\partial \phi} = -\frac{\frac{\partial F}{\partial \phi}}{\frac{\partial F}{\partial \Lambda}}.
\end{equation}

From Proposition \ref{prop:existenceLambda}, \begin{equation} \label{eq:F_lambda_neg}
    \frac{\partial F}{\partial \Lambda} \bigg|_{\Lambda=\Lambda^*} < 0.
\end{equation}

Consequently, the sign of $\frac{\partial \Lambda^*}{\partial \phi}$ in Equation \eqref{eq:implicit_F} is determined by the sign of $\frac{\partial F}{\partial \phi}$,
\begin{equation}\label{eq:dfdphi}
    \frac{\partial F}{\partial \phi} =  \underbrace{\frac{\beta_v s\epsilon}{(\delta + s\epsilon)\gamma_v} \left(1 - e^{-(\delta + s\epsilon)\Lambda^*}\right)}_{V(\Lambda^*)} - \underbrace{\frac{\beta_u}{\gamma_u} \left(1 - e^{-s\Lambda^*}\right)}_{U(\Lambda^*)}.
\end{equation}

We can rewrite $U(\Lambda^*)$ and $V(\Lambda^*)$ in terms of $h(x)$:
\begin{equation}  
    U(\Lambda^*) = \frac{\beta_u}{\gamma_u} s \cdot h(s) , \quad V(\Lambda^*) = \frac{\beta_v}{\gamma_v} s\epsilon \cdot h(\delta + s\epsilon),
\end{equation}
where $h(x) := \frac{1 - e^{-x\Lambda^*}}{x}$. Its derivative is:
\begin{equation}
    h'(x) = \frac{e^{-x\Lambda^*}(1 + x\Lambda^*) - 1}{x^2}.
\end{equation}
Since $e^y > 1 + y$ for all $y > 0$, it follows that $1 > e^{-y}(1 + y)$, making the numerator strictly negative. Therefore, $h'(x) < 0$, establishing that $h(x)$ is strictly decreasing for $x > 0$.

For the high boosting condition, $\delta > \delta^*$, it follows that $\delta + s\epsilon > s$. Therefore, because $h(x)$ is strictly decreasing, $h(\delta + s\epsilon) < h(s)$. Given \eqref{eq:vaccineassumption} and $\epsilon \in (0,1)$,
\begin{equation}
    V(\Lambda^*)= \frac{\beta_v}{\gamma_v} s\epsilon \cdot h(\delta + s\epsilon) < \frac{\beta_u}{\gamma_u} s \cdot h(\delta + s\epsilon) < \frac{\beta_u}{\gamma_u} s \cdot h(s) = U(\Lambda^*).
\end{equation}
Since $V(\Lambda^*) < U(\Lambda^*)$, it follows that $\frac{\partial F}{\partial \phi} < 0$ for $\delta>\delta^*$. Note that in the frequent boosting regime, the proof still holds for a net-beneficial vaccine \eqref{eq:net_beneficial}.

If $\delta < \delta^*$, it follows that $\delta + s\epsilon < s$. Assuming $\frac{\beta_v}{\gamma_v}\leq \frac{\beta_u}{\gamma_u}$, and given that $\frac{s\epsilon}{\delta+s\epsilon}<1$ in the presence of boosting ($\delta>0$), Equation \eqref{eq:dfdphi} implies that $\frac{\partial F}{\partial \phi}<0$.
This completes the proof.

\section{Proof of Proposition \ref{prop:VEoevrall}}\label{sec:VEoverallApp}

First, note that $AR_u$ and $AR_v$ strictly increase with $\Lambda$ by Equations \eqref{eq:finalsizeULambda} and \eqref{eq:finalsizeVLambda}, respectively. Therefore, by Equation \eqref{eq:Z}, the final epidemic size, $Z$, also strictly increases with $\Lambda$. By Proposition \ref{prop:vaccination_boosting}, if \eqref{eq:vaccineassumption} holds, $\Lambda'(\phi) < 0$ for all boosting parameters $\delta > 0$. Consequently, $Z$ is a strictly decreasing function of $\phi$, meaning $Z'(\phi) < 0$. Differentiating ${\rm VE}_{\rm overall}$ with respect to $\phi$ yields
\begin{equation*}
    {\rm VE}_{\rm overall}'(\phi) = -\frac{Z'(\phi)}{Z(0)}.
\end{equation*}
Since $Z(0) > 0$ and $Z'(\phi) < 0$, it follows that ${\rm VE}_{\rm overall}'(\phi) > 0$, and thus ${\rm VE}_{\rm overall}$ increases with the vaccine coverage, $\phi$.
\section{Numerical Simulations of Net-Beneficial Vaccines}\label{sec:appendixnumericalsimulations}
We consider cases involving risk compensation, where $\frac{\beta_v}{\gamma_v}>\frac{\beta_u}{\gamma_u}$, while ensuring the vaccine remains net-beneficial \eqref{eq:net_beneficial}. We demonstrate numerically that in the infrequent boosting regime ($\delta < \delta^*$), the behavior of ${\rm VE}$ with respect to vaccine coverage $\phi$ varies: it can either decrease and then increase (Figure \ref{fig:MastercodeApp}A), or strictly decrease (Figure \ref{fig:MastercodeApp}B). This contrasts with the baseline scenario, $\beta_u = \beta_v$, shown in Figure \ref{fig:MastercodeVE}B, where ${\rm VE}$ strictly increases with $\phi$ in the infrequent boosting regime.

\begin{figure}[!htb]
    \centering
    \includegraphics[width=.99\linewidth]{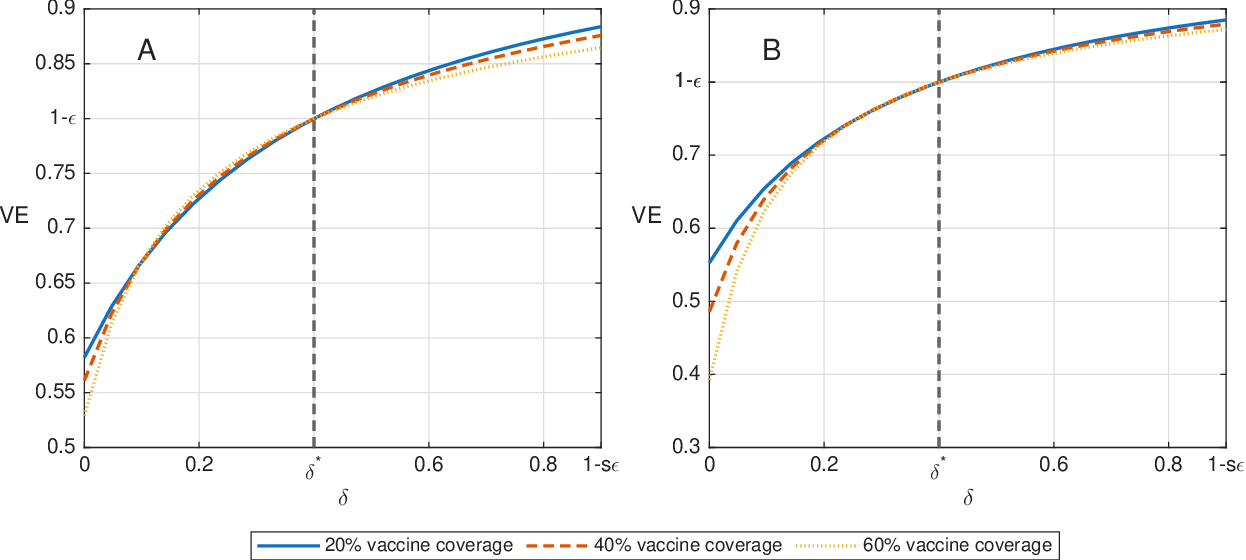}
    \caption[Vaccine effectiveness calculated as a function of $\delta$ for different values of vaccine coverage]
    { Vaccine effectiveness (${\rm VE}$) calculated as a function of the boosting parameter $\delta$ across three vaccine coverage levels, $\phi = 0.2, 0.4, 0.6$. Common parameters for both panels are $\epsilon = 0.2$, $\beta_u = 5$, $\gamma_u = \gamma_v = 1$, and $s = 0.5$. {\bf A.} For $\beta_v = 15$, in the regime where $\delta < \delta^*$, ${\rm VE}$ initially decreases with vaccine coverage and then increases. {\bf B.} For $\beta_v = 20$, in the same $\delta < \delta^*$ regime, ${\rm VE}$ strictly decreases as vaccine coverage increases.}
    \label{fig:MastercodeApp}
\end{figure}

\section{Proof of Proposition \ref{prop:ARasymptotitcLambda}}\label{sec:highRvLambdaApp}

    We prove that when $\beta_u\to\infty$ or $\beta_v\to\infty$ the positive solution $\Lambda$ of equation \eqref{eq:finalsizelambda} is bounded below by a positive constant. This implies that when $\beta_u\to\infty$ or $\beta_v\to\infty$, $\Lambda\to\infty$.  
   \newtheorem{lemma}{Lemma}
\begin{lemma}

\label{lemma:lowerboundLambda}
    If 
    \begin{equation}\label{eq:conditionbeta}
        \frac{\beta_u(1-\phi)}{\gamma_u}\left(1 - e^{-\frac{1}{2}s}\right) + \frac{\beta_v \phi s \epsilon}{(\delta+s\epsilon) \gamma_v}\left(1 - e^{-\frac{1}{2}(\delta+s\epsilon)}\right)-\frac{1}{2}>0,
    \end{equation} and $\mathcal{R}_v>1$,
     then the solution $\Lambda>0$ of Equation \eqref{eq:finalsizelambda} satisfies $\Lambda > \frac{1}{2}$.

\end{lemma}
\begin{proof}[Proof of Lemma \ref{lemma:lowerboundLambda}]
Assume condition \eqref{eq:conditionbeta} holds, and $\mathcal{R}_v>1$, and let $\Lambda^{*}$ denote the unique positive solution to Equation \eqref{eq:finalsizelambda}, as guaranteed by Proposition \ref{prop:existenceLambda}. By definition, $\Lambda^*$ is a root of the function $F(\Lambda)$, as defined in Equation \eqref{eq:defoff}, such that $F(\Lambda^*) = 0$.

By \eqref{eq:conditionbeta}, $F$ at $\Lambda = 1/2$ satisfies:
\begin{equation}
   F\left( \frac{1}{2} \right)= \frac{\beta_u(1-\phi)}{\gamma_u}\left(1 - e^{-\frac{1}{2}s}\right) + \frac{\beta_v \phi s \epsilon}{(\delta+s\epsilon) \gamma_v}\left(1 - e^{-\frac{1}{2}(\delta+s\epsilon)}\right)-\frac{1}{2}>0. \label{eq:F(Lambda)_eval}
\end{equation}

Because $F(0) = 0$, $F(1/2) > 0$, and $F$ is strictly concave, the function cannot cross the horizontal axis before $\Lambda = 1/2$. Therefore, the unique positive root must satisfy $\Lambda^* > 1/2$.
\end{proof}
\begin{proof}[Proof of Proposition \ref{prop:ARasymptotitcLambda}]
    Now, from Equation \eqref{eq:finalsizelambda} and Lemma \ref{lemma:lowerboundLambda},
\begin{equation}\label{eq:limLambda}
    \lim_{\beta_u\to \infty} \Lambda\geq \lim_{\beta_u \to \infty} \frac{\beta_u(1-\phi)}{\gamma_u}\left(1 - e^{-\frac{1}{2}s}\right) + \frac{\beta_v \phi s \epsilon}{(\delta+s\epsilon) \gamma_v}\left(1 - e^{-\frac{1}{2}(\delta+s\epsilon)}\right) = \infty,
\end{equation}
and similarly,
\begin{equation}\label{eq:limLambdaV}
    \lim_{\beta_v\to \infty} \Lambda\geq \lim_{\beta_v \to \infty} \frac{\beta_u(1-\phi)}{\gamma_u}\left(1 - e^{-\frac{1}{2}s}\right) + \frac{\beta_v \phi s \epsilon}{(\delta+s\epsilon) \gamma_v}\left(1 - e^{-\frac{1}{2}(\delta+s\epsilon)}\right) = \infty.
\end{equation}
Therefore, from Equations \eqref{eq:finalsizeULambda}, \eqref{eq:finalsizeVLambda}, \eqref{eq:finalsizelambda}, \eqref{eq:limLambda}, \eqref{eq:limLambdaV}, for $\phi>0$, we conclude:
\begin{equation}\label{eq:ARlargeRvApp}
    \begin{split}
     &\lim_{\beta_u\to \infty} AR_u=\lim_{\beta_v\to \infty} AR_u=\lim_{\Lambda\to \infty} AR_u=1,\\
     &\lim_{\beta_u\to \infty} AR_v=\lim_{\beta_v\to \infty} AR_v=\lim_{\Lambda\to \infty} AR_v=\frac{s\epsilon}{\delta+s\epsilon}.
    \end{split}  
\end{equation}
\end{proof}

\section{Proofs of Propositions \ref{prop:VEmmonotnedelta} and \ref{prop:VEmmonotneinn}}\label{sec:VEintuitionAppendics}
\VEmmonotnedelta*

\begin{proof}

 By Equation \eqref{eq:VEncalc}, this is a direct consequence of the fact that the term $\frac{1-(1-\delta-s\epsilon)^n}{\delta+s\epsilon} = \sum_{i=0}^{n-1} (1-\delta-s\epsilon)^i$  is strictly decreasing in $\delta$.
\end{proof}
\VEmmonotneinn*

\begin{proof}
By \eqref{eq:Qnudirect}, \eqref{eq:Qnvdirect} we get:
    \begin{equation}\label{eq:qnv/qnu}
        \frac{Q_n^{v}}{Q_n^{u}}=\frac{(1-(1-\delta-s\epsilon)^{n})s\epsilon}{(1-(1-s)^{n})(\delta+s\epsilon)}.
    \end{equation}
    For convenience, denote $x=1-\delta-s\epsilon$, $y=1-s$.
    Substituting \[\frac{1-(1-\delta-s\epsilon)^{n}}{\delta+s\epsilon}=\frac{1-x^{n}}{1-x}= S_n(x),\] and \[\frac{1-(1-s)^{n}}{s}=\frac{1-y^{n}}{1-y}= S_n(y),\] into \eqref{eq:qnv/qnu} yields,
\begin{equation}
         \frac{Q_n^{v}}{Q_n^{u}}=\epsilon R_n,
    \end{equation}
where we denote \[R_n=\frac{S_n(x)}{S_n(y)}.\]
   
    To complete the proof, we need to prove that when $x>y$, $R_{n+1}>R_n$, and when $x<y$, $R_{n+1}<R_n$. That is equivalent to showing, 
    \begin{equation}
        \begin{split}
           & \forall x>y: \frac{S_{n+1}(x)}{S_n(x)}>\frac{S_{n+1}(y)}{S_n(y)},\\
           & \forall x<y: \frac{S_{n+1}(x)}{S_n(x)}<\frac{S_{n+1}(y)}{S_n(y)}. 
        \end{split}
    \end{equation}
Let
    \begin{equation}
        f_n(t)=\frac{S_{n+1}(t)}{S_n(t)}.
    \end{equation} Thus, we need to show that $f_n(t)$ is monotone increasing for all $t\in [0,1)$. 
    We have:
    \begin{equation}
        f_n(t)=\frac{t^{n+1}-1}{t^{n}-1}.
    \end{equation}
    We differentiate and obtain: 
    \begin{equation}
        f_n^{\prime}(t)=\frac{t^{n-1}[t^{n+1}-(n+1)t+n]}{(t^{n}-1)^{2}}.
    \end{equation}
    Since $t^{n-1},(t^{n}-1)^{2}\geq0$, we prove that for $t\in [0,1)$, 
    \begin{equation}
        P_n(t)=t^{n+1}-(n+1)t+n\geq0.
    \end{equation}
    Indeed, since 
    \begin{equation}
         P_n^{\prime}(t)=(n+1)t^{n}-(n+1), \quad P_n^{''}= n(n+1)t^{n-1}\geq0,
    \end{equation}
    $P_n(t)$ has a unique global minimum in $t=1$. Therefore $P_n(t)\geq P_n(1)=0$ for all $t\in [0,1)$.

As for the third part of the proposition, substituting $\delta=\delta^{*}=s(1-\epsilon)$ into \eqref{eq:VEncalc}, we obtain:
\begin{equation}
    {\rm VE}_n=1-\frac{s\epsilon[1-(1-s)^{n}]}{s[1-(1-s)^n]}=1-\epsilon.
\end{equation}
\end{proof}

\end{appendices}



\end{document}